\PassOptionsToPackage{hidelinks}{hyperref}
\documentclass[manuscript,screen]{acmart}

\usepackage{soul}

\usepackage{comment}
\usepackage{amssymb}
\usepackage{amsmath}
\usepackage{algorithm,algpseudocode}
\usepackage{multicol}
\usepackage{multirow}
\usepackage{soul}
\usepackage{lipsum}

\usepackage{wrapfig}
\usepackage{subfigure}
\newtheorem{thm}{Theorem}

\newtheorem{defn}[thm]{Definition}

\newtheorem{prob}[thm]{Problem}
\usepackage{todonotes}
\usepackage{braket}
\input{Qcircuit}
\usepackage{glossaries}
\newacronym{bsm}{BSM}{Bell-State Measurement}
\newacronym{css}{CSS}{Calderbank-Shor-Steane}

\usepackage[usestackEOL]{stackengine}
\edef\tmp{\the\baselineskip}
\setstackgap{L}{\tmp}

\AtBeginDocument{%
  \providecommand\BibTeX{{%
    \normalfont B\kern-0.5em{\scshape i\kern-0.25em b}\kern-0.8em\TeX}}}

\setcopyright{acmcopyright}
\copyrightyear{2023}
\acmDOI{10.1145/1122445.1122456}

\acmBooktitle{Woodstock '18: ACM Symposium on Neural Gaze Detection,
  June 03--05, 2018, Woodstock, NY}
\acmPrice{15.00}
\acmISBN{978-1-4503-XXXX-X/18/06}

\begin{document}

\title{Repeated Purification versus Concatenated Error Correction in Fault Tolerant Quantum Networks}

\author{Michel Barbeau}
\authornotemark[1]
\email{barbeau@scs.carleron.ca}
\author{Joaquin Garcia-Alfaro}
\authornotemark[2]
\email{garcia\_a@telecom-sudparis.eu}
\author{Evangelos Kranakis}
\authornotemark[1]
\email{kranakis@scs.carleron.ca}
\affiliation{%
  \institution{\\ $^*$School of Computer Science, Carleton University}
  \streetaddress{1125 Col. By Dr.}
  \city{Ottawa}
  \state{Ontario}
  \country{Canada}
  \postcode{K1S 5B6}
}
\affiliation{%
  \institution{$^\dagger$SAMOVAR, T\'el\'ecom SudParis, Institut Polytechnique de Paris}
  \streetaddress{19 place Marguerite Perey}
  \city{Palaiseau}
  \country{France}
  \postcode{91120}
}

\renewcommand{\shortauthors}{Barbeau, Garcia-Alfaro, Kranakis}
\renewcommand{\shorttitle}{Quantum Networks and Work Memory Requirements}

\begin{abstract}
Entanglement distribution is a core mechanism for the future quantum Internet.
The quantum world is, however, a faulty environment.
Hence, successful entanglement swapping is error-prone.
The occurrence of quantum state errors can be mitigated using purification and error correction,
which can be repeated in the former case and concatenated in the latter case.
Repeated purification merges low-fidelity qubits into higher-quality ones, 
while concatenated error correction builds upon the redundancy of quantum information.
In this article, we study in-depth and compare the two options: repeated purification and concatenated error correction.
We consider using repeated purification and concatenated error correction to mitigate the presence of faults that occur during the establishment of Bell pairs between remote network nodes.
We compare their performance versus the number of repetitions or concatenations, 
to reach a certain level of fidelity in quantum networks.
We study their resource requirements, namely, their work memory complexity (e.g., number of stored qubits) and operational complexity (e.g., number of operations).
Our analysis demonstrates that concatenated error correction, versus repeated purification, 
requires fewer iterations and has lower operational complexity than repeated purification to reach high fidelity at the expense of increased memory requirements. 
\end{abstract}


\maketitle

\section{Introduction}

The promise of quantum computing raises doubts and questions  
in some~\cite{Dyakonov2019} and provokes enthusiasm in others~\cite{Zhong1460}. 
An essential component of the future quantum world is the quantum Internet, 
which will transmit quantum data through a network~\cite{gyongyosi2022advances}.
A core function required by the quantum Internet to achieve this task is entanglement distribution,
which aims at creating Bell pairs between distant nodes.
These Bell pairs can then transfer quantum data, using teleportation, or classical data, using super-dense coding from one node to the other.
However, it is widely acknowledged that quantum computing and networking are faulty environments with a high likelihood of errors in quantum states.
The rate of errors can be reduced using purification and error correction.
The quality of a quantum state can be improved by repeating the former multiple times or concatenating several instances of the latter.
Both purification and error correction deserve to be considered to handle the occurrence of quantum state faults. 
Repeated purification enhances the quality of a quantum state by merging low-fidelity qubits into higher-quality ones. 
For concatenated error correction, the quality enhancement procedure leverages redundant quantum information. 
Recent research shows the relevance of quantum error correction to make advances in the quantum Internet~\cite{gyongyosi2022advances,biercuk2022quantum}. 
Software tools that benefit from such techniques include security applications (e.g., enforcement of quantum key establishment~\cite{scarani2009security}) and distributed applications (e.g., augmenting the parallelism of quantum machine learning~\cite{biamonte2017quantum}).

An interesting question is which, among repeated purification and concatenated error correction, is best to use in the quantum Internet.
In this article, we analyze this question in depth.
For each case, namely, repeated purification and concatenated error correction, we develop analytic models of fidelity, 
quantum memory complexity (number of required qubits), and time complexity (number of operations).
We quantify the requirements (memory and operations) to improve the quality of a quantum state using repeated purification and concatenated error correction. The former includes work memory (e.g., number of qubits stored), while the latter refers to the process's fidelity improvement operational complexity (e.g., number of steps required).
Our analysis demonstrates that reaching high fidelity via concatenated error correction requires fewer iterations than repeated purification at the expense of increased memory requirements. 
At the same time, repeated purification increases the operational complexity in contrast with concatenated error correction. 
Numeric simulations complement the results obtained with the analytic models.

The article is structured as follows. Section~\ref{sec:related-work} surveys related work. 
Section~\ref{sec:netmodel} presents our model for quantum networks and repeaters. 
Section~\ref{sec:entanglement-swapping} introduces entanglement swapping and discusses our error model and repeated purification in conjunction with concatenated error correction for fidelity enhancement.  
Section~\ref{sec:efficiency-analysis} provides our analysis of the efficiency of each approach (i.e., repeated purification vs. concatenated error correction).  
Section~\ref{sec:generalgraphs} presents the results of numeric simulations. 
Finally, Section~\ref{sec:conc} summarizes our main work and discusses prospects for future research.

\section{Related Work}
\label{sec:related-work}

Metropolitan-scale applications using quantum networks have been recently analyzed via simulation~\cite{yehia2022quantum,coopmans2021netsquid}. Results show 
that practical quantum-enhanced network functionalities may soon be ready for security applications (e.g., enforcement of quantum key establishment~\cite{scarani2009security}) and distributed applications (e.g., augmenting the parallelism of quantum machine learning~\cite{biamonte2017quantum}). Nevertheless, quantum-enhanced networks must face the problem of entanglement distribution under fidelity constraints~\cite{stephens2013hybrid}. Given a source and a destination, the length of the channel imposes a significant decay in the quality of communication. Several protocols have been proposed to address this problem~\cite{Pirandola2019,Bauml2018,Gyongyosi2018,Pant2019}. The first step relies on finding a path of quantum repeaters capable of conducting end-to-end entanglement distribution via entanglement swapping~\cite{Zukowski1993,Caleffi2017,perseguers2010,Schoute2018}. Such repeaters are assumed to be imperfect~\cite{Briegel1998}, i.e., the execution of Bell pair creation and entanglement swapping may fail. Hence, entanglement distribution is traditionally modeled as a probabilistic process~\cite{9351761}. To maximize the likelihood of success, the process can be scheduled, e.g., following a sequential or nested model~\cite{dai2020optimal}.

This article assumes an entanglement swapping model with Bell pairs in homogeneous repeater chains~\cite{dur2007entanglement}. Decay of communication is considered under the bit-flip error 
model (i.e., Pauli $X$ errors). We consider using quantum memories that can store qubits for short periods of time~\cite{Azuma2015,Munro2015,liu2017semihierarchical} to handle failures in the chain of homogeneous repeaters.  Resource requirements, including work memory (e.g., number of stored qubits) and operational complexity (e.g., number of transformations), to attain a certain level of fidelity, are investigated. The amount of quantum memory required by repeaters performing entanglement distribution can be quantified to achieve a target fidelity~\cite{barbeau2020capacity}. There are works with a focus on linear~\cite{Jiang2007,Shchukin2019,Caleffi2017}, grid~\cite{Pant2019}, ring~\cite{Schoute2018,Shirichian2018} or  sphere~\cite{Schoute2018} topologies, in which the size of the required quantum memory is related to the number of neighbors that a repeater has. Under this classical error model, repeated purification~\cite{zhou2016purification} vs. concatenated repetition codes~\cite{fan2022entanglement} can be compared~\cite{dur2007entanglement} to decide which one achieves better performance with minimal increase in resources. Our analysis demonstrates that concatenated error correction requires fewer iterations and operations than repeated purification at the expense of increased memory requirements. 

\section{Quantum Networking}\label{sec:netmodel}

We begin with a description of our quantum network model and architecture. As network model~\cite{peterson2007computer},
let us consider a connected graph denoted as the pair $G= (V, E)$, where $V$ is the vertex set and $E$ the edge set. Assume that the set $V$ of vertices is partitioned into two distinct subsets $R$ and $T$, such that $V= R \cup T$, where
$R$ is the set of {\em repeaters} and
$T$ is the set of {\em terminals}.
A network example is depicted in Fig.~\ref{fig:qua1}.
An edge represents a bi-directional quantum channel that can be used to establish entanglement between its two endpoints.
All repeaters and terminals can perform entanglement swapping, but have limited quantum memory. Let us consider a path $p= v_0, v_1 , \ldots , v_{n-1}, v_n$ of $n+1$ vertices in $V$.
The start and final vertices $v_0$ and $v_n$ are terminals in $T$.
All intermediate vertices $v_i \in R$, $0 < i < n$, are repeaters in $R$. The number $n$ of edges of $p$ is its length and is denoted by $\vert p \vert$.

\begin{figure}[!hptb]
	\begin{center}
		\includegraphics[width=0.5\columnwidth]{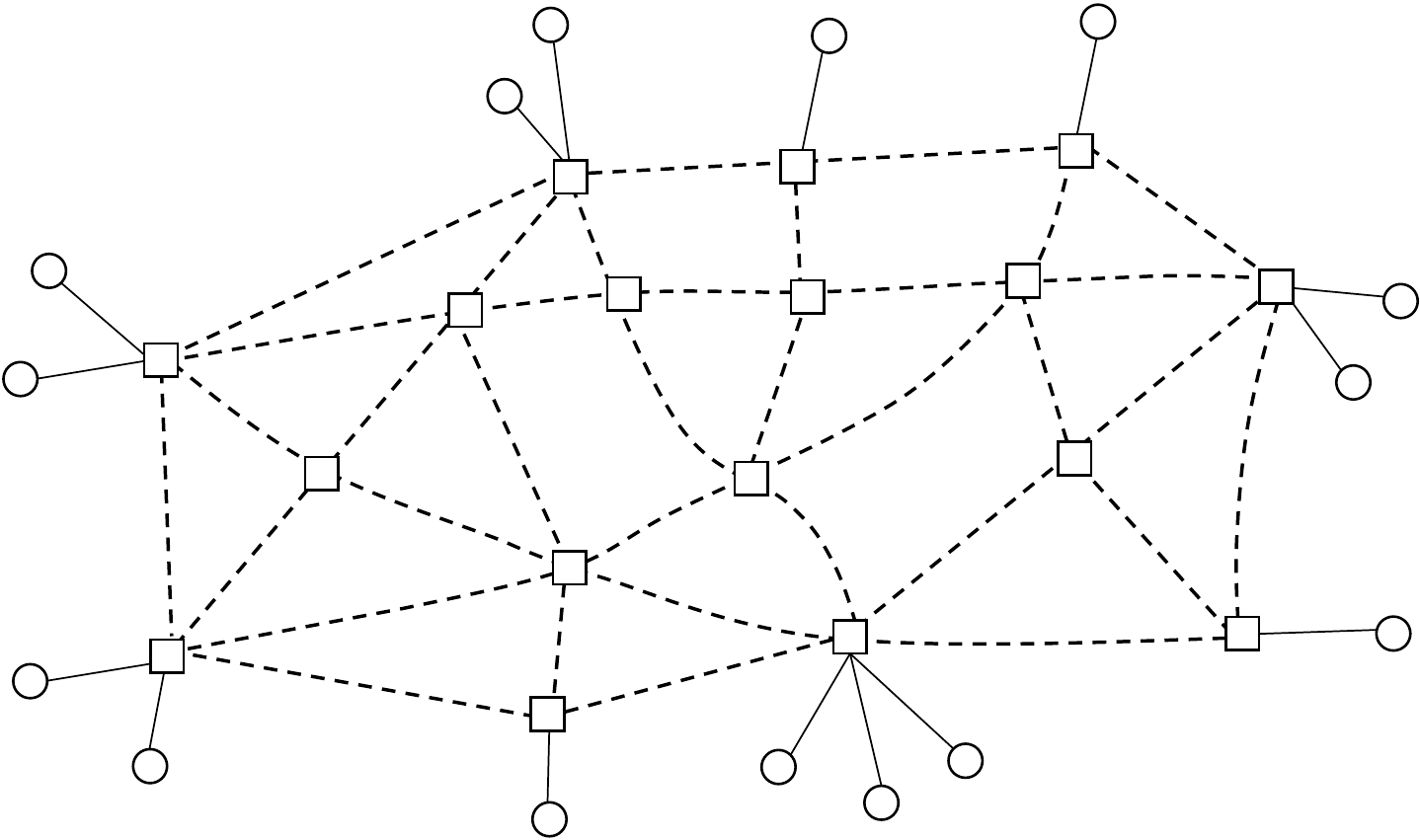}
	\end{center}
	\caption{Vertices depicted by disks are terminals in $T$ while those depicted by squares are repeaters in $R$. Note that terminals are directly connected to repeaters. A dashed line represents a path consisting of repeaters; the endpoints of the path connect terminals.\label{fig:qua1}}	
\end{figure}

\begin{defn}[Complete Set of Paths]
	A set $P$ of paths in the given graph $G$ is called {\em complete} when for any pair $t , t'$ of terminals in $T$ there is a path $p= v_0, v_1 , \ldots , v_{n-1}, v_n$ in $P$ such that $t= v_0, t'= v_n$ and all intermediate vertices $v_i$, for $0 < i < n$, are repeaters in $R$.
\end{defn}

\noindent We require that: 
i)~terminal nodes are not adjacent to each other,
ii)~every terminal node is adjacent to at least one repeater, and
iii)~adjacent repeaters can communicate directly with each other.

\begin{figure}[!hptb]
	\begin{center}
		\includegraphics[width=0.5\columnwidth]{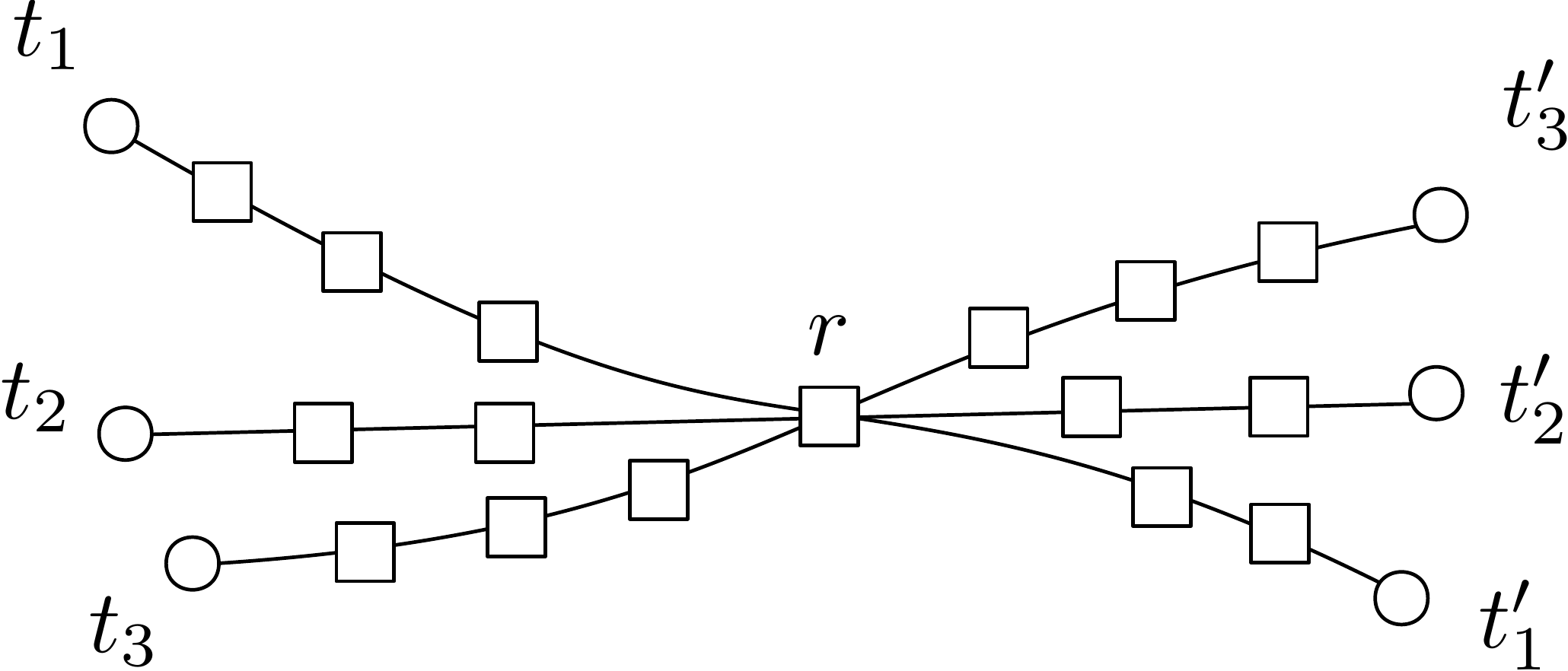}
	\end{center}
	\caption{The capacity of a repeater $r$ is the number of paths that go through $r$; in the picture, this is equal to three. Note that the hollow squares depict repeaters (from the set $R$) and hollow disks terminals (from the set $T$).}
	\label{fig:graph1}
\end{figure}
\begin{defn}[Capacity of a Repeater]\label{def:capacity}
	For a given complete set $P$ of paths in a graph $G$ and a repeater $r \in R$, 
	let the capacity of $r$, denoted by $C_P (r)$, be defined as the number of paths $p\in P$ that pass through $r$, see Fig.~\ref{fig:graph1}.
\end{defn}
According to this definition, the capacity of a repeater is proportional to the number of qubits it must be able to store, paired with entanglement, to enable communication between terminals. 

\begin{defn}[Capacity Induced by a Complete Set of Paths] For a given complete set $P$ of paths in a graph $G$, the capacity of $G$ induced by the complete set $P$ of paths is defined as the maximum capacity caused by repeaters in $R$. It is defined by the formula
	\begin{equation}
	\label{congest0:eq}
	C_G (P) := \max_{r \in R} C_P ( r ) \mbox{ qubits.} 
	\end{equation}
\end{defn}

Let ${\mathcal P}_G$ denote the set of complete sets of paths for the graph $G$. 
When this is understood from the context, we omit the subscript $G$ in ${\mathcal P}_G$. Among the collections of a complete set of paths 
for the graph $G$, we are interested in minimizing the quantity $C_G (P)$, namely 
\begin{equation}
\label{congest1:eq}
\min_{P \in {\mathcal P}_G} C_G (P) = \min_{P \in {\mathcal P}_G} \max_{r \in R} C_P ( r ) \mbox{ qubits},
\end{equation}
where the minimum is taken over the set ${\mathcal P}_G$ of all possible complete sets of paths $P$ for the graph $G$. 

\begin{figure}[!b]
	\begin{center}
		\includegraphics[width=0.65\columnwidth]{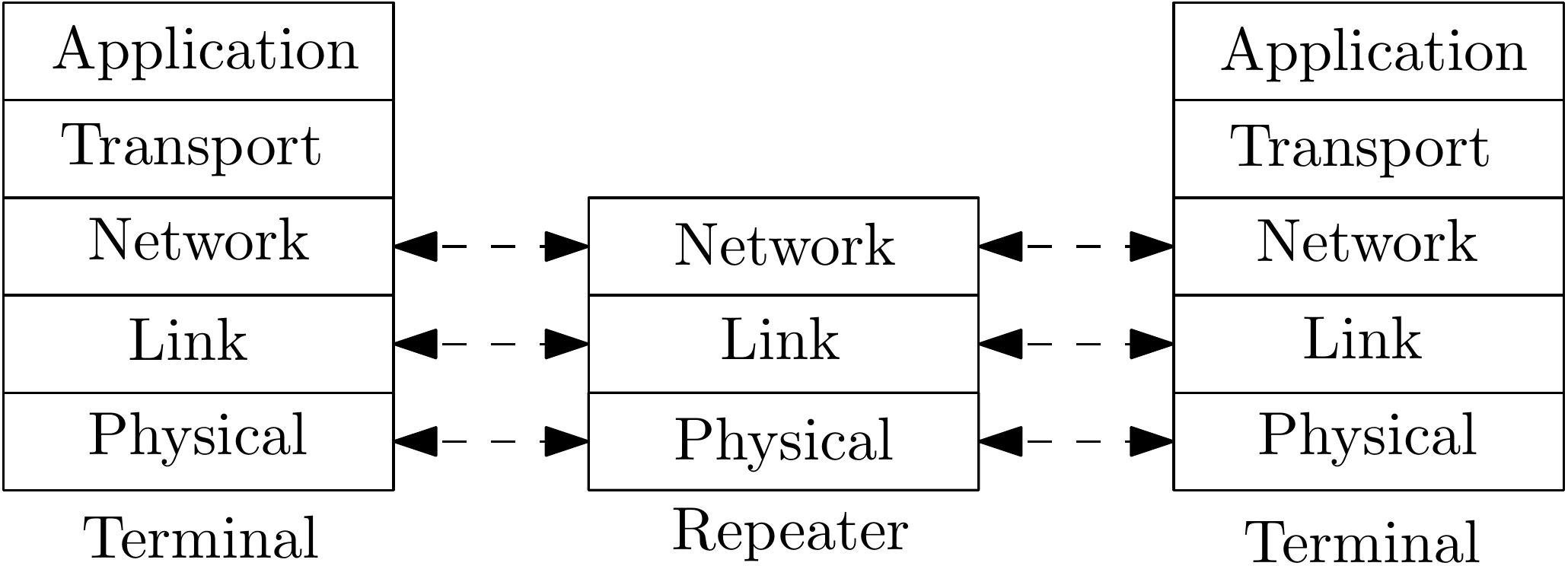}
	\end{center}
	\caption{Quantum network layered architecture.} 
	\label{fig:networkarchitecture}
\end{figure}

The main problem concerns the capacity induced by connecting all pairs of terminals by paths consisting of repeaters. We are aiming for an algorithm that defines the set $P$ of paths
while minimizing the resulting capacity induced on the graph. The main problem is formally described as follows.
\begin{prob}
	Given a graph $G$, find a complete set  $P$ of paths that attains or approximates the quantity $\min_{P \in {\mathcal P}_G} C_G (P) $.
\end{prob}
This capacity is 
somewhat related to the number of qubits the repeaters would need to store to allow communication between terminals using qubits.
In the sequel, we explore these various aspects of quantum network architecture.

Quantum networks can be visualized by employing layering architectures~\cite{matsuo2019quantum}. 
Figure~\ref{fig:networkarchitecture} depicts 
the layers assumed in this work.
Each network node (repeaters and terminals) comprises a physical, a link, and a network layer. The physical layer assumes error-prone point-to-point transfer of quantum bits, using single photons or laser pulses. It also assumes the generation of low-fidelity entangled pairs using, e.g., entangled photon pair source devices~\cite{jones2016design}. Such devices can be configured to attempt the creation of Bell pairs continuously. When Bell pair creation is successful, the resulting two qubits are physically separated. For instance, the first qubit can go to the left node and the second to the right node. The physical transmission of entangled pairs is also prone to errors.  Non-perfect entanglement (fidelity below one) assumes Bell pair creation, but potentially contains errors. Fidelity improvement methods can be conducted at the upper layers (at the link and network layers). These two upper layers implement purification, error correction, and entanglement swapping. There is a transport layer and an application layer in all the terminals. The transport layer uses the network layer to provide end-to-end transfer of quantum states to participating processes. The transport layer may also implement end-to-end error correction. The application layer comprises processes running quantum algorithms.

\section{Fault Tolerant Quantum Networking}
\label{sec:entanglement-swapping}

The quantum Internet will rely on establishing quality Bell pairs between network nodes.
Bell pairs can transfer quantum states using teleportation, or classical data, using super-dense coding.
For two neighbor network nodes connected by a direct link, Bell pairs can be established by leveraging parametric down-conversion.
Bell pairs can be established with entanglement swapping for remote network nodes, not directly linked.
Nevertheless, both parametric down-conversion and entanglement swapping may be faulty and cause errors. 
In the sequel, we assume that adjacent nodes, repeaters, or terminals,
use direct communications to establish Bell pairs.
Quantum repeaters and entanglement swapping establish Bell pairs between non-adjacent repeaters and terminals.
Bell pair establishment procedures can be faulty and introduce errors.
Fault tolerance can be achieved using purification and error correction.
The usage of entanglement swapping, repeated purification, and concatenated error correction is discussed further.

Quantum errors resulting from faulty procedures
are modeled in several ways.
For the sake of simplicity and without loss of generality,
let us consider a bit-flip model.
Randomly, qubit $\ket{0}$ is converted to $\ket{1}$, or vice versa. 
An error changes the qubit $\alpha\ket{0} + \beta\ket{1}$ to the qubit $\beta\ket{0} + \alpha\ket{1}$.
For a Bell pair such as
\begin{equation}
\label{eq:purification:bellpair}
\ket{\Phi^+}  =
\frac{\Ket{00} + \Ket{11}}{\sqrt{2}}
\end{equation}
when both qubits are inverted, 
the errors are canceled, the term $\Ket{00}$ becomes $\Ket{11}$ and vice versa.
However, the presence of a single qubit error results in the quantum state
\begin{equation}
\label{eq:purification:bellpairwitherror}
\ket{\Psi^+} = \frac{\Ket{01} + \Ket{10}}{\sqrt{2}}.
\end{equation}
In this example, the first or second qubit flips, but not both.
However, in both cases, the outcome is the same.
When the qubit in the first position is flipped,
the term $\Ket{00}$ is transformed to $\Ket{10}$ and
the term $\Ket{11}$ is transformed $\Ket{01}$. 
When the qubit in the second position is flipped,
the term $\Ket{00}$ is transformed to $\Ket{01}$, and
the term $\Ket{11}$ is transformed to $\Ket{10}$.
The resulting quantum state is the same for both events, Eq.~\eqref{eq:purification:bellpairwitherror}.

Let $p\in [0,1]$ be the probability of a single qubit inversion error in a Bell pair.
The error model is represented as the quantum state
\begin{equation}
\label{eq:purification:errormodel}
\sqrt{1-p} \ket{\Phi^+}
+
\sqrt{p} \ket{\Psi^+}.
\end{equation}

Entanglement swapping is a core quantum networking procedure.
In an instance of the procedure,
three network nodes are involved: a source $s$, a repeater $r$, and a destination $d$.
The source and destination can be repeaters or terminals.
In the sequel, possession subscripts are used.
The subscript $s$ in the ket-notation $\ket{\phi}_s$ means that the qubit $\ket{\phi}$ is possessed by nodes $s$.
The subscript $sd$ in the ket-notation $\ket{\Phi^+}_{s d}$ means that the Bell pair $\ket{\Phi^+}$ is shared between nodes $s$ and $d$. 
Node $s$ possesses the first qubit, while node $d$ possesses the second qubit.
Entanglement swapping assumes that $s, r$ and $r,d$ 
share the Bell pairs $\ket{\Phi^+}_{s r}$ and $\ket{\Phi^+}_{r d}$, respectively.
They may have been created using parametric down-conversion or previous instances of entanglement swapping.
Next, the repeater $r$ does Bell measurement of the second qubit of $\ket{\Phi^+}_{s r}$ and the first qubit of $\ket{\Phi^+}_{r d}$ into the classical bits $c_1$ and $c_2$, respectively.
The repeater sends the two classical bits $c_1$ and $c_2$ resulting from the measurement to the destination $d$.
If $c_2$ is equal to one, then $d$ applies the Pauli gate $X$ to the second qubit of $\ket{\Phi^+}_{r d}$.
If $c_1$ is equal to one, then $d$ also applies to gate $Z$.
The final result is an end-to-end Bell pair $\ket{\Phi^+}_{s d}$ shared between the source and
destination, 
the first qubit of $\ket{\Phi^+}_{s r}$ and second qubit of $\ket{\Phi^+}_{r d}$, 
possibly transformed by the $X$ and $Z$ gates.

\begin{figure*}[!htb]
	    \includegraphics[width=11cm]{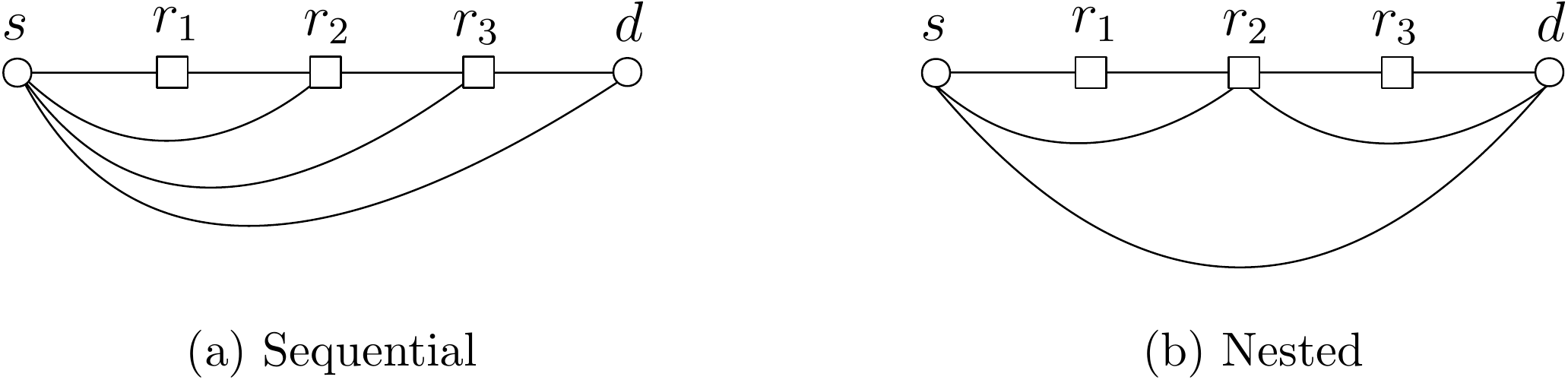}
	\caption{Scheduled swapping, (a) sequential and (b) nested.}
	\label{fig:swapping}
\end{figure*}
Entanglement swapping can be multi-hop.
In such a scenario,
several swapping operations are used to establish entanglement between two distant terminals, $s$ and $d$,
connected by a multi-hop path.
Using a routing algorithm,  
a path $p= r_0=s, r_1 , \ldots , r_{n-1}, r_n=d$ is chosen,
where $r_1 , \ldots , r_{n-1}$ are repeaters, 
with $r_{0}$ equal to $s$ and $r_{n}$ equal to $d$.
Entanglement is established stage by stage.
There are two available scheduled swapping protocols~\cite{dai2020optimal,9351761}, namely, sequential and nested, 
see Figure~\ref{fig:swapping}.
For the sequential protocol, for $i=1,2,3,\ldots,n-1$,
using repeater $r_{i}$ as intermediate, an entanglement swapping operation creates a Bell pair $\ket{\Phi^+}_{s r_{i+1}}$ between nodes
$s$ and $r_{i+1}$.
In $n-1$ iterations, 
a Bell pair $\ket{\Phi^+}_{s d}$ is  created between 
terminals $s$ and $d$.

For the sake of simplicity, let us assume that the path length $n$ is a power of two.
With the nested protocol, firstly, for $i=0,2,4,\ldots,n-2$,
using repeater $r_{i+1}$ as intermediate, an entanglement swapping operation creates a Bell pair $\ket{\Phi^+}_{r_{i} r_{i+2}}$ between nodes
$r_{i}$ and $r_{i+2}$.
Next, for $i=0,4,8,\ldots,n-4$,
using node repeater $r_{i+2}$ as intermediate, 
entanglement swapping creates a Bell pair $\ket{\Phi^+}_{r_{i} r_{i+4}}$ between nodes $r_{i}$ and $r_{i+4}$.
Each iteration doubles the length of the segment bridged by a pair.
In $\log_2 n$ iterations, 
a Bell pair $\ket{\Phi^+}_{s d}$ is  created between 
terminals $s$ and $d$.
Note that in both protocols, sequential and nested, 
every repeater needs to be able to store two qubits.

\begin{figure*}[!htb]
    \includegraphics[width=5cm]{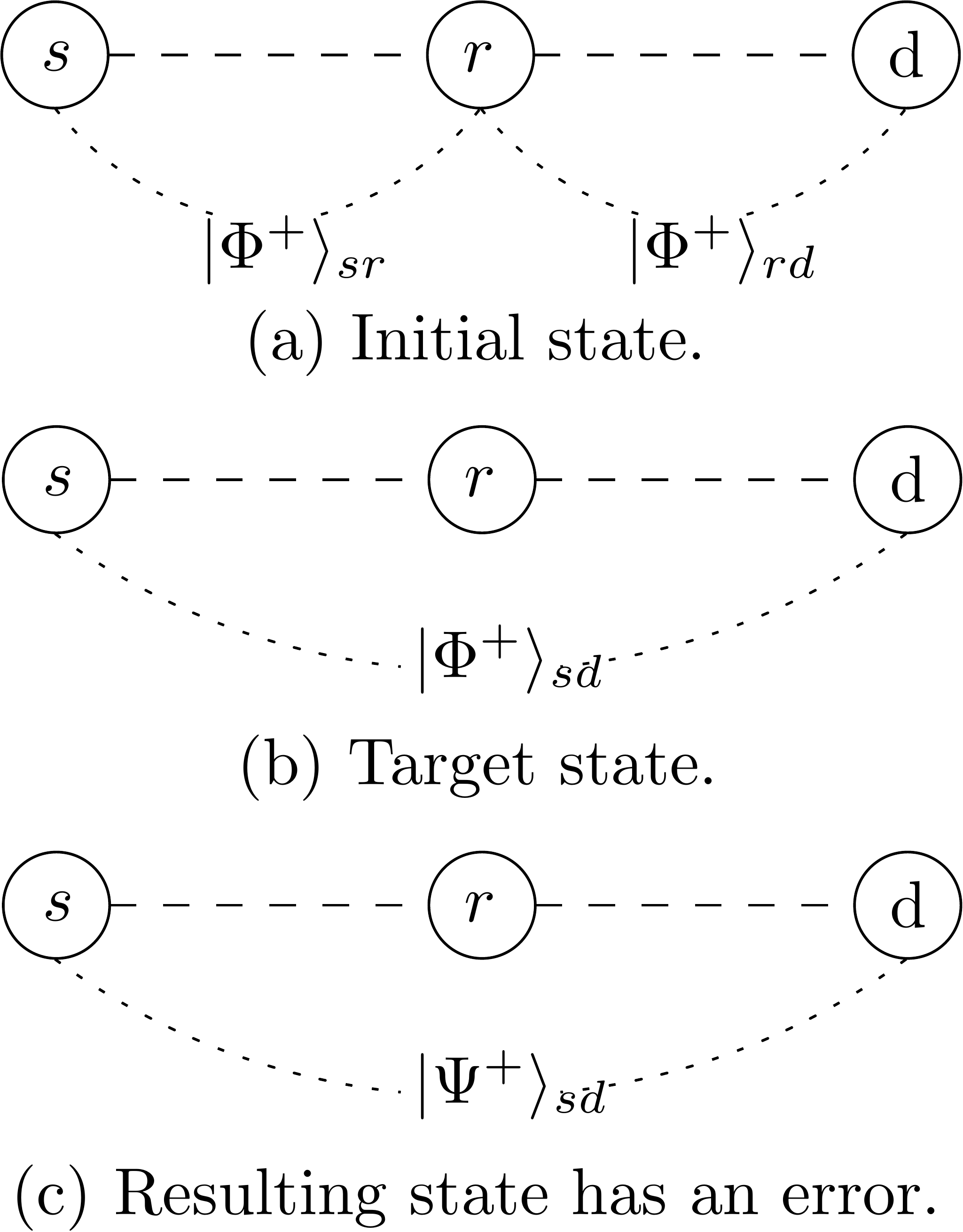}
	\caption{Entanglement swapping in the presence of a bit-flip error. (a)~Initial state. (b)~Target outcome. (c)~Faulty outcome.}
	\label{fig:swapping-n-error}
\end{figure*}
Errors may result from faulty entanglement swapping.
Figure~\ref{fig:swapping-n-error} (a) pictures the initial state of an instance of the entanglement swapping procedure.
There is a chain of two point-to-point links connecting terminal $s$ to repeater $r$ and repeater $r$ to terminal $d$.
Terminal $s$ shares a Bell pair $\ket{\Phi^+}_{s r}$ with repeater $r$.
Repeater $r$ shares a Bell pair $\ket{\Phi^+}_{r s}$ with terminal $s$.
Part (b) shows the target state resulting from a successful entanglement swapping. 
Terminal $s$ shares a Bell pair $\ket{\Phi^+}_{s d}$ with terminal $d$.
Assuming the bit flip error model, Part~(c) shows an entanglement swapping that failed to produce a correct result.
Terminal $s$ shares a Bell pair $\ket{\Psi^+}$ with terminal $d$.
Qubit errors are an important quantum networking problem.
However, the quality of a Bell pair can be improved with repeated purification and  concatenated error correction,
discussed in the sequel.

Purification is a procedure executed between two nodes, 
a source terminal $s$ and a destination terminal $d$, 
to augment the fidelity of Bell pairs~\cite{van2014quantum}. 
The outcome of purification can be characterized by the probability to correct errors successfully, 
referring to the concept of fidelity.
Fidelity indicates the degree of resemblance of a quantum state to its original value.
Fidelity is affected by errors but can be improved using purification.
The goal is to establish high-fidelity Bell pairs between quantum network nodes.

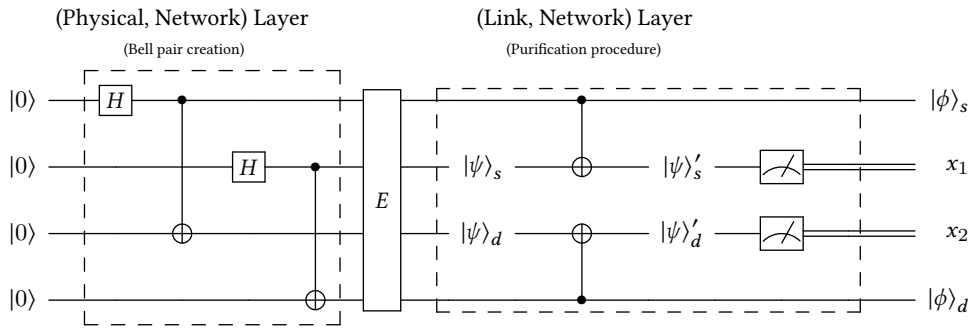
\begin{figure}[!b]
	\begin{center}
		\[
    \Qcircuit @C=1em @R=1em @! {	  
    & &  \mbox{\begin{tabular}{cc} (Physical, Network) Layer  \\ \scriptsize{~~~~~~~~~~(Bell pair creation)} \end{tabular}} & & & & & & 
    \mbox{\begin{tabular}{cc} ~~~~~~~~~(Link, Network) Layer\\\scriptsize{~~~~~~~~~(Purification procedure)}\end{tabular}} \\
    \lstick{\ket{0}} & \gate{H} & \ctrl{2} & \qw & \qw & \multigate{3}{E}       & \qw & \qw                   & \ctrl{1} & \qw    & \qw                    & \qw      & \qw    &  \qw & \lstick{\ket{\phi}_s} \\
    \lstick{\ket{0}} & \qw      & \qw    & \gate{H}  & \ctrl{2} & \ghost{E}        & \qw & \lstick{\ket{\psi}_s}     & \targ      & \qw & \lstick{\ket{\psi}'_s}  & \meter & \cw & \cw & \lstick{x_1} \\
    \lstick{\ket{0}} & \qw & \targ & \qw & \qw & \ghost{E}                    & \qw & \lstick{\ket{\psi}_d}  & \targ    & \qw & \lstick{\ket{\psi}'_d}   & \meter  &  \cw & \cw & \lstick{x_2} \\
    \lstick{\ket{0}} & \qw      & \qw    & \qw & \targ & \ghost{E}                   & \qw &  \qw  & \ctrl{-1}   & \qw  & \qw &  \qw &  \qw & \qw & \lstick{\ket{\phi}_d} 
  \gategroup{2}{2}{5}{5}{1.25em}{--}
  \gategroup{2}{7}{5}{13}{1em}{--}
    }
    \]
	\caption{Purification procedure. According to the bit-flip error model, gate $E$ represents the arbitrary corruption of qubits. As depicted, Bell pair creation may happen at the physical layer, followed by the link layer; it may also happen at the network layer.}\label{fig:purificaton}
	\end{center}
\end{figure}

Figure~\ref{fig:purificaton} depicts purification.
The goal is to establish the high-fidelity Bell pair $\ket{\Phi^+}_{s d}$. 
The leftmost dotted-line rectangle represents the attempt by a repeater $r$ to create two Bell pairs
shared with nodes $s$ and $d$.
During Bell pair establishment, 
qubit errors can be introduced.
Rectangle $E$ models the introduction of errors.
For the sake of simplicity, let us consider solely bit-flip errors.
The action of $E$ on every one of the four qubits is defined as the following weighted sum:
\begin{equation}
\label{eq:bitflipmodel}
\sigma = \sqrt{p_I} \cdot I + \sqrt{p_X} \cdot X, \mbox{ with identity and Pauli matrices $I$ and $X$, and probabilities } 
p_I + p_X = 1
\end{equation}
The term $p_X$ is a bit-flip error probability. 
Gate $E$ is the tensor product of four such gates, that is, $E= \sigma^{\otimes 4}$.
In the output of $E$,
let us denote the first pair's first qubit as $\ket{\phi}_s$ while the first qubit of the second pair is $\ket{\psi}_s$, 
both possessed by terminal $s$.
Using $\ket{\phi}_s$ as the control qubit and $\ket{\psi}_s$ the target qubit,
node $s$ applies a $CNOT$ gate yielding the pair
$\ket{\phi}_s\ket{\psi'}_s = CNOT \left( \ket{\phi}_s\ket{\psi}_s \right)$.

Let us denote the first pair's second qubit as $\ket{\phi}_d$ while the second qubit of the second pair is $\ket{\psi}_d$, 
both possessed by terminal $d$.
Using $\ket{\phi}_d$ as the control qubit and $\ket{\phi}_d$ as the target qubit,
node $d$ applies a $CNOT$ gate yielding the pair
$\ket{\phi}_d\ket{\psi}'_d = CNOT \left( \ket{\phi}_d\ket{\psi}_d \right)$.
Terminal $s$ measures $\ket{\psi}'_s$ into classical bit $x_1$.
Terminal $d$ measures $\ket{\psi}'_d$ into classical bit $x_2$.
Using classical communications, 
$s$ and $d$ compare the values of $x_1$ and $x_2$.
When they are equal, it is concluded that
the pair $\ket{\phi}_s \ket{\phi}_d$
has been purified and corresponds to the Bell pair $\ket{\Phi^+}_{sd}$.
When $x_1$ and $x_2$ are different, it is concluded that the pair $\ket{\phi}_s\ket{\phi}_d$, is not equal to $\ket{\Phi^+}_{sd}$. 
Purification failed. 
The pair $\ket{\phi}_s\ket{\phi}_d$, is rejected.

There are four possible purification outcomes.
When $x_1$ is equal to $x_2$, there are two cases.
There are no errors, and both pairs $\ket{\phi}_s\ket{\phi}_d$ or the pair $\ket{\psi}_s\ket{\psi}_d$ at the output of gate $E$ are in state $\ket{\Phi^+}_{sd}$.
Or, both pairs are in error in state $\ket{\Psi^+}_{sd}$.
The probability of the first case is $q^2$, with $q=1-p$. 
Purification is successful.
The probability of the second case is $p^2$.
The errors are undetected and purification fails
and wrongly concludes with success.
When $x_1$ and $x_2$ are different, 
there is a qubit error in one of the pairs.
Either the pair $\ket{\phi}_s\ket{\phi}_d$ or the pair $\ket{\psi}_s\ket{\psi}_d$ is in state $\ket{\Psi^+}$,
but not both.
The probability for the first or second pair to be in error is $qp$.
In both cases, the error is detected.
Purification fails.

In this setting, as a function of the single qubit inversion error probability $p$ in a Bell pair,
fidelity becomes equivalent to the likelihood of the absence
of errors when purification concludes with a positive result, that is:
\begin{equation}
\label{eq:purification:fidelity}
f(p) = \frac{q^2}{q^2 + p^2} \mbox{ with $q=1-p$} .
\end{equation}

An important observation is that purification does not improve the  fidelity, i.e., the condition $f(p) > q$,
when the input fidelity is $0.5$ or below (see Fig.~9.2 in Ref.~\cite{van2014quantum}). 
A desired degree of fidelity, approaching one, can be obtained with several purification rounds,
see Section~\ref{sec:efficiency-analysis-Pu}.


An alternative to purification is error correction.
\begin{table}[!hbt]
\caption{Error correction codes. In the $(n,k)$ notation,
the variables $n$ and $k$ indicate the number of physical qubits and a corresponding number of logical qubits used to encode them. Available choices for $n$ or $k$ are listed, separated by vertical bars.}\label{tab:correctioncode}
\begin{center}
\begin{tabular}{ll}
\hline
Code & Example $(n,k)$  \\ 
\hline
Calderbank-Shor-Steane (CSS)~\cite{calderbank1996good,steane1996error} & $(5 \vert 7 \vert 9,1)$ \\
Hadamard-Steane & $(7,3 \vert 4)$ \\
Repetition & $(3, 1)$\\
Shor & $(9,1)$\\
Steane~\cite{steane1996error}       & $(5 \vert 7,1)$  \\
&  
\end{tabular}
\end{center}
\end{table}
Table~\ref{tab:correctioncode} lists error correction codes by name.
A distinction is made between physical qubits and logical qubits.
All error correction codes use several physical qubits to represent every abstract logical qubit.
A pair $(n,k)$ is associated with every correction code.
Parameter $n$ represents the number of physical qubits used to encode $k$ logical qubits.

Desurvire's analysis (Ref.~\cite{desurvire2009classical} Sec.~24.1) for the bit-flip model shows that a $(3,1)$ error correction code
improves fidelity when the bit-flip probability is greater than or equal to 0.5. 
Longer codes ($n>3$) do not yield better results because they increase the risks of getting more errors.
Error correction can be recursively applied, or concatenated, several times~\cite{knill1996concatenated,fan2022entanglement}.
This is further studied in Section~\ref{sec:efficiency-analysis-EC}.

\begin{figure}[h]
	\begin{center}
		\[
		\Qcircuit @C=1em @R=1em @! {	 
			& & \mbox{\begin{tabular}{cc}Physical Layer\\\scriptsize{(Bell pair creation)}\end{tabular}}  &  & & & &  & & \mbox{\begin{tabular}{cc}Link Layer \\\scriptsize{(Encoding \& Decoding)}\end{tabular}}\\
			\lstick{\ket{0}} & \gate{H} & \ctrl{3} & \qw & \lstick{\ket{\phi}_1} & \ctrl{1}  & \ctrl{2} & \multigate{5}{E} &  \qw & \lstick{\ket{\phi}'_s} &  \ctrl{2} & \ctrl{1} & \targ   &  \qw  &  \qw &  \lstick{\ket{\phi}''_s}\\
			\lstick{\ket{0}} & \qw   &  \qw &  \qw & \qw & \targ & \qw &  \ghost{E} &  \qw &  \qw & \qw & \targ &  \ctrl{-1} & \qw  & \qw  \\
			\lstick{\ket{0}} & \qw  &  \qw &  \qw & \qw & \qw & \targ & \ghost{E} &  \qw & \qw & \targ & \qw & \ctrl{-2} &  \qw & \qw  \\
			\lstick{\ket{0}} & \qw   & \targ  & \qw &  \lstick{\ket{\phi}_2} &  \ctrl{1}  & \ctrl{2} &  \ghost{E} & \qw & \lstick{\ket{\phi}'_d} &  \ctrl{2} & \ctrl{1} & \targ   &  \qw  &  \qw &  \lstick{\ket{\phi}''_d}\\
			\lstick{\ket{0}} &  \qw & \qw & \qw &  \qw & \targ & \qw  & \ghost{E} &  \qw &  \qw & \qw &  \targ & \ctrl{-1} & \qw  & \qw  \\
			\lstick{\ket{0}} &  \qw &  \qw & \qw &  \qw & \qw & \targ &  \ghost{E} & \qw & \qw &  \targ & \qw & \ctrl{-2} &  \qw & \qw  
			\gategroup{2}{2}{7}{5}{0.95em}{--}
			\gategroup{2}{6}{7}{14}{1.7em}{--}
        	}
		\]
		\caption{Encoding and decoding procedure. Gate $E$ arbitrarily corrupts qubits according to the bit-flip error model. The physical layer creates the Bell pair, while the link layer maps qubits into code words. In addition, the link layer receives the qubit triples and decodes them back to single qubits.}
		\label{fig:error-correction-circuit-3}
	\end{center}
\end{figure}
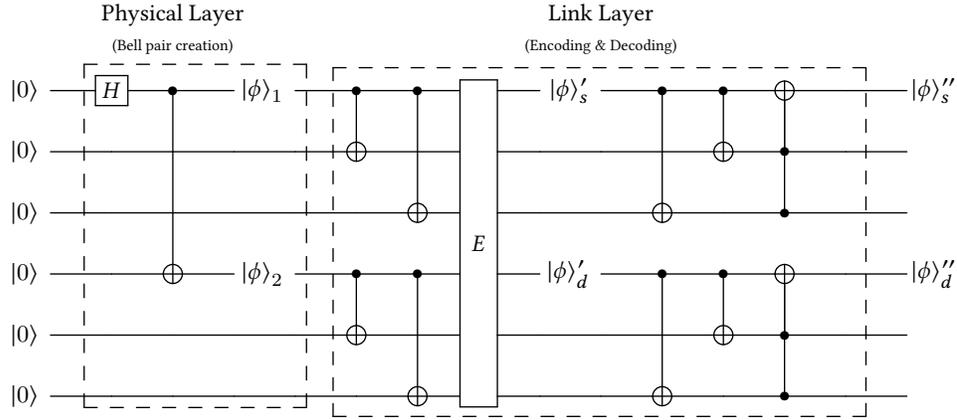
Figure~\ref{fig:error-correction-circuit-3} depicts the circuit associated with a single qubit error correction procedure, as
used in the sequel. It represents repeater $r$ and two terminals $s$ and $d$.
A $(3,1)$ repetition code is used. The leftmost rectangle represents the behavior of the repeater.
The input qubits $\ket{0}_1$ and $\ket{0}_2$ are the Bell pair selectors. Since they are both equal to $\ket{0}$, 
the $H$ gate and first $CNOT$ gate of $r$ create the Bell pair $\ket{\Phi^+}$.
The second and third $CNOT$ gates of $r$ map the two members of the pair to a specific
codeword, for instance,
$\ket{0}$ is mapped $\ket{000}$ and $\ket{1}$ is mapped $\ket{111}$. %
The quantum channel $E$ can arbitrarily corrupt qubits and introduce bit-flip errors.
The gate $E$ is the tensor product of six copies of the gate defined in Equation~\eqref{eq:bitflipmodel}, 
that is, $E= \sigma^{\otimes 6}$.
The behavior of every terminal is identical.
Terminal $s$ ($d$) receives a three-qubit code word.
The first and second $CNOT$ gates map a code word to a single qubit (first line),
that is, $\ket{000}$ is mapped $\ket{000}$ and $\ket{111}$ is mapped $\ket{100}$.
The Toffoli gate performs single qubit error correction.
Both qubits on the second and third lines need to be $\ket{1}$ to flip the qubit's value on the target on the first line.
For example, if the input to the terminal is the three qubits $\ket{011}$,
it decodes into $\ket{011}$.
The first qubit is in error.
The Toffoli gate flips the first qubit into $\ket{1}$. 

Entanglement swapping can be used in conjunction with error correction. 
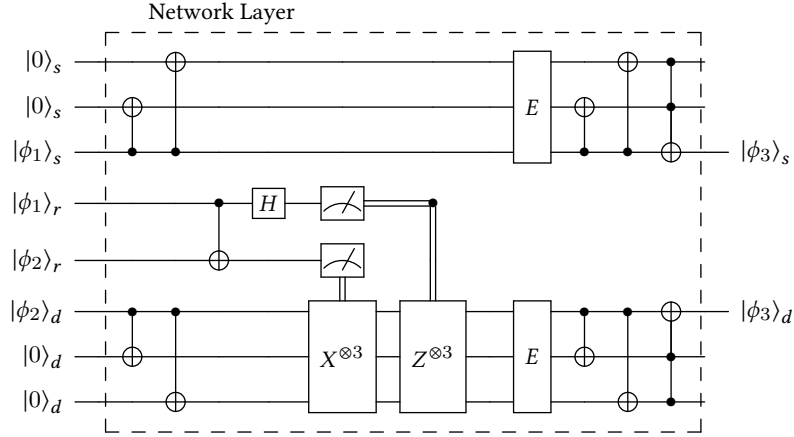
\begin{figure}[!htb]
	\begin{center}
		\[
		\Qcircuit @C=1em @R=1em {	 
			& & & &  \mbox{\begin{tabular}{cc}~~~~~Network Layer\\\scriptsize{~~~}\end{tabular}} \\
            \lstick{\ket{0}_s}      & \qw   & \qw       & \targ     & \qw & \qw & \qw & \qw                     & \qw & \multigate{2}{E} & \qw       & \targ     & \ctrl{2} & \qw \\
			\lstick{\ket{0}_s}      & \qw   & \targ     & \qw       & \qw & \qw & \qw & \qw                     & \qw & \ghost{E}         & \targ     & \qw       & \ctrl{1} & \qw  \\
			\lstick{\ket{\phi_1}_s} & \qw   & \ctrl{-1} & \ctrl{-2} & \qw & \qw & \qw & \qw                     & \qw & \ghost{E}         & \ctrl{-1} & \ctrl{-2} & \targ & \qw & \rstick{\ket{\phi_3}_s} \qw \\
			\lstick{\ket{\phi_1}_r} & \qw   & \qw       & \qw       & \ctrl{1} & \gate{H} & \meter & \control \cw \\
            \lstick{\ket{\phi_2}_r} & \qw   & \qw       & \qw       & \targ    & \qw  & \meter & \cwx \\
   			\lstick{\ket{\phi_2}_d} & \qw   & \ctrl{1}  & \ctrl{2}  & \qw      & \qw  & \multigate{2}{X^{\otimes 3}}\cwx    & \multigate{2}{Z^{\otimes 3}}\cwx & \qw & \multigate{2}{E} & \ctrl{1}  & \ctrl{2} & \targ & \qw & \rstick{\ket{\phi_3}_d} \qw \\
   			\lstick{\ket{0}_d}      & \qw   & \targ     & \qw       & \qw      & \qw  & \ghost{X^{\otimes 3}}    & \ghost{Z^{\otimes 3}}        & \qw & \ghost{E} & \targ     & \qw      & \ctrl{-1} & \qw \\
			\lstick{\ket{0}_d}      & \qw   & \qw       & \targ     & \qw      & \qw  & \ghost{X^{\otimes 3}}    & \ghost{Z^{\otimes 3}}        & \qw & \ghost{E} & \qw       & \targ    & \ctrl{-2} & \qw 
            \gategroup{2}{3}{9}{13}{2.25em}{--}
		   }
		\]
		\caption{Entanglement swapping procedure with bit-flip errors and $(3,1)$ repetition error correction. Gate $E$ arbitrarily flips qubits, $E = \sigma^{\otimes 3}$.}
		\label{fig:swapping-error-correction-circuit}
	\end{center}
\end{figure}
Figure~\ref{fig:swapping-error-correction-circuit} depicts a circuit describing a qubit-flip error correction procedure combined with entanglement swapping.
It represents repeater $r$ and two terminals $s$ and $d$.
A $(3,1)$ repetition code is used for error correction.
The entanglement-swapping procedure is under the control of repeater $r$.
The repeater $r$ entangles remote three-qubit code words, 
contrasting with base entanglement swapping that entangles remote individual qubits.
The subscripts $s$, $r$, and $d$ are used to denote qubit possession by the nodes $s$, $r$, and $d$, respectively. 
The qubits $\ket{0}_s$ and $\ket{0}_d$ are ancillary.
The first and second qubits of the Bell pair $\ket{\Phi^+}_{sr}$ shared by terminal $s$ and repeater $r$ are denoted as $\ket{\phi_1}_{s}$ and $\ket{\phi_1}_{r}$, respectively.
The  first and second qubits of the Bell pair $\ket{\Phi^+}_{rd}$ shared by repeater $r$ and terminal $d$ are denoted as $\ket{\phi_2}_{r}$ and $\ket{\phi_2}_{d}$, respectively.

Before entanglement swapping starts, the physical qubits $\ket{\phi_1}_{s}$ and $\ket{\phi_2}_{d}$ are encoded into logical qubits.
The second and third $CNOT$ gates of $s$ map qubits $\ket{\phi_1}_{s}$ to three qubit code words, that is,
$\ket{0}$ is mapped to $\ket{000}$ and $\ket{1}$ is mapped to $\ket{111}$.
Similarly, the second and third $CNOT$ gates of $d$ map qubits $\ket{\phi_2}_{d}$ to three qubit code words, that is,
$\ket{0}$ is mapped to $\ket{000}$ and $\ket{1}$ is mapped to $\ket{111}$.

The entanglement-swapping procedure introduces bit-flip errors,
modeled by the $E$ gates.
Following entanglement swapping with errors, the behavior of every terminal is similar.
Every terminal, $s$ and $d$, possesses a three-qubit block.
They apply consecutively two $CNOT$ gates and a Toffoli gate.
They correct single bit-flip errors.

In the absence of errors, the final state is
\begin{equation*}
\frac{
\ket{000}_s\ket{000}_d + \ket{111}_s\ket{111}_d
}
{\sqrt{2}}.
\end{equation*}
Projecting on the third qubit (possessed by $s$) and fourth qubit (possessed by $d$), 
it corresponds to the shared Bell pair $\ket{\Phi^+}_{sd}$,
whose first qubit is denoted as $\ket{\phi_3}_s$ and second qubit as $\ket{\phi_3}_d$.
In case of a single-bit flip error, at $s$, for example, on the third line,
the final state is
\begin{equation*}
\frac{
\ket{110}_s\ket{000}_d + \ket{111}_s\ket{111}_d
}
{\sqrt{2}}.
\end{equation*}

Projecting on the third qubit (possessed by $s$) and fourth qubit (possessed by $d$), 
it also yields the shared Bell pair $\ket{\Phi^+}_{sd}$.
Double or triple bit-flip errors at node $s$ or node $d$ are not corrected.

In the next lemma, the fidelity of a Bell pair corrected with a $(3,1)$ repetition code is characterized. 

\begin{lemma}
\label{lm:efficacy}
The fidelity of a Bell pair error-corrected with the $(3,1)$ repetition code is 
$\mathcal{F}(p)= (1 -p)^3 + 3p(1 - p)^2$.   
\end{lemma}

\begin{proof}
Considering that a Bell pair involves two qubits, the proof follows from an analysis in Ref.~\cite{desurvire2009classical}. 
It states that given the bit-flip error probability $p$, where $p$ is smaller than $1/2$, the fidelity of a single qubit corrected with a $(3,1)$ repetition code is $F(p)=\sqrt{(1 -p)^3 + 3p(1 - p)^2}$. 
Hence, 
the fidelity of the Bell pair is $\mathcal{F}(p) = F(p)^2 = (1 -p)^3 + 3p(1 - p)^2$.
\end{proof}

\begin{figure}[!htb]
\begin{center}
\[
\Qcircuit @C=1em @R=1em {	
\lstick{\ket{0}_s}      & \qw       & \qw       & \qw       & \targ     & \qw                     & \qw & \multigate{8}{E} & \qw       & \targ     & \ctrl{2} & \qw       & \qw       & \qw       & \qw \\
\lstick{\ket{0}_s}      & \qw       & \qw       & \targ     & \qw       & \qw                     & \qw & \ghost{E}        & \targ     & \qw       & \ctrl{1} & \qw       & \qw       & \qw       & \qw \\
\lstick{\ket{0}_s}      & \qw       & \targ     & \ctrl{-1} & \ctrl{-2} & \qw                     & \qw & \ghost{E}        & \ctrl{-1} & \ctrl{-2} & \targ    & \qw       & \targ     & \ctrl{6}  & \qw \\
\lstick{\ket{0}_s}      & \qw       & \qw       & \qw       & \targ     & \qw                     & \qw & \ghost{E}        & \qw       & \targ     & \ctrl{2} & \qw       & \qw       & \qw       & \qw \\
\lstick{\ket{0}_s}      & \qw       & \qw       & \targ     & \qw       & \qw                     & \qw & \ghost{E}        & \targ     & \qw       & \ctrl{1} & \qw       & \qw       & \qw       & \qw \\
\lstick{\ket{0}_s}      & \targ     & \qw       & \ctrl{-1} & \ctrl{-2} & \qw                     & \qw & \ghost{E}        & \ctrl{-1} & \ctrl{-2} & \targ    & \targ     & \qw       & \ctrl{3}  & \qw \\
\lstick{\ket{0}_s}      & \qw       & \qw       & \qw       & \targ     & \qw                     & \qw & \ghost{E}        & \qw       & \targ     & \ctrl{2} & \qw       & \qw       & \qw       & \qw \\
\lstick{\ket{0}_s}      & \qw       & \qw       & \targ     & \qw       & \qw                     & \qw & \ghost{E}        & \targ     & \qw       & \ctrl{1} & \qw       & \qw       & \qw       & \qw \\
\lstick{\ket{\phi_1}_s} & \ctrl{-3} & \ctrl{-6} & \ctrl{-1} & \ctrl{-2} & \multigate{0}{Swapping} & \qw & \ghost{E}        & \ctrl{-1} & \ctrl{-2} & \targ    & \ctrl{-3} & \ctrl{-6} & \targ     & \rstick{\ket{\phi_3}_s}
}
\]
\caption{Entanglement swapping procedure with concatenated $(3,1)$ repetition error correction. Only the circuit at the location of terminal $s$ is shown. A similar circuit for terminal $d$ can be drawn. Gate $E$ arbitrarily flips qubits, $E = \sigma^{\otimes 9}$.}
\label{fig:swapping-error-correction-repetition}
\end{center}
\end{figure}
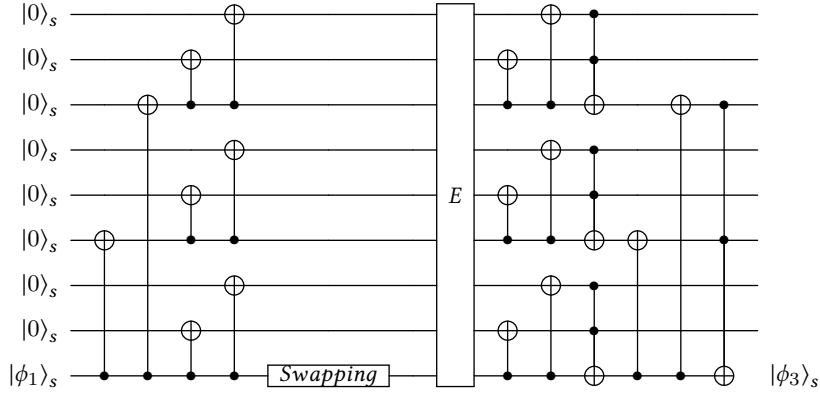
Error correction can also be iterated.
Figure~\ref{fig:swapping-error-correction-repetition} shows a two-concatenation scenario, focusing on the qubit $\ket{\phi_1}_s$ of terminal $s$.
The $(3,1)$ repetition error correction is used.
The qubit $\ket{\phi_1}_s$ is mapped to a three-qubit codeword, the first, fourth, and seventh lines from the bottom.
Each qubit of this codeword is recursively mapped to a three-qubit low-level codeword.
Following the swapping procedure, gate $E$ models the occurence of by flip-errors. 
Then, every low-level three-qubit block is mapped to a three-qubit block, 
which is in turn mapped to the single qubit $\ket{\phi_3}_s$.
Using this model, error correction can be applied several times recursively to improve the quality of a qubit involved in an entanglement swapping procedure on every side.
This is analyzed in detail in the following section.

\section{Analysis of Repeated Purification and Concatenated Error Correction}
\label{sec:efficiency-analysis}

The fidelity of a Bell pair may be increased with repeated purification and concatenated error correction.
This capability is analyzed in detail in this section.
The precise requirements for repeated purification or concatenated error correction are identified and compared against each other.

\subsection{Purification Analysis}
\label{sec:efficiency-analysis-Pu}

Let $F$ be the initial fidelity of a Bell pair, with $F > 1/2$.
Using repeated purification,
let us consider the fidelity sequence $F_n$ defined recursively as follows:
\begin{equation}
\label{eq:purent4}
F_0 = F \mbox{ and } F_{n} = \frac{F_{n-1}^2}{F_{n-1}^2 + (1 - F_{n-1} )^2} \mbox{ for } n=1,2,3,\ldots
\end{equation}
Observe that the sequence $F_n$ is increasing. Therefore the limit $\lambda:= \lim_{n \to \infty} F_n$ exists. Passing to the limit in the righthand side of Equation~\eqref{eq:purent4}, we conclude that $\lambda =1$. This follows from the fact that
$
\lambda = \frac{\lambda^2}{\lambda^2 + (1 - \lambda )^2}, 
$
which gives the solution $\lambda = 1$. Therefore taking into account that $F_n \to 1$ as $n \to \infty$, we can make repeated purifications to the resulting state 
\begin{equation}
\label{eq:purent6}
\rho_n = F_n \Ket{\Phi^{+}} \Bra{\Phi^{+}} + (1-F_n ) \Ket{\Psi^{+}} \Bra{\Psi^{+}}.
\end{equation}

To get an idea of the costs required to obtain a certain fidelity, 
the speed of convergence of repeated purification is investigated. 
The main question of interest is the following:
\begin{quote}
Given a starting value $F_0 > 1/2$, how many purification repetitions are needed to obtain a certain level of fidelity? 
Namely, given an arbitrarily small  $\epsilon > 0$ for what value of $n$ can we claim $F_n > 1 -\epsilon$? 
\end{quote}
The following analysis of our main question requires some claims we state and prove below as lemmas. 

Consider the operation $F \to F' := \frac{F^2}{F^2 + (1 - F )^2}$. 
We look at the ratio of the new value $F'$ versus the old $F$, namely
$
\frac{F'}{F} = \frac{\frac{F^2}{F^2 + (1 - F )^2} }{F} = \frac{F}{F^2 + (1 - F )^2} 
$
and prove the following lemma.

\begin{lemma}
\label{lm:speed1}
If $1/2 < F < 1/{\sqrt{2}}$ then 
\begin{equation}
\label{eq:speed1}
\frac{F'}{F} > 1+ \left(F - \frac 12\right) = \frac 12 + F.
\end{equation}
\end{lemma}
\begin{proof}
Multiply Inequality~\eqref{eq:speed1} by two to obtain $\frac{2F}{F^2 + (1 - F )^2} >1 + 2F$. Then multiply out to derive the equivalent form $2F > (F^2 + (1 - F )^2) (1+2F)$.
When we expand and simplify the last inequality, it turns out to be equivalent to $(2F-1)(2F^2 -1) < 0$. Since by assumption $1/2 < F$, 
we conclude that the last inequality is equivalent to $F < 1/{\sqrt{2}}$, which is valid from the hypothesis of the lemma. 
\end{proof}

The next observation is that the purification operation is monotone increasing.
\begin{lemma}
\label{lm:speed0}
The purification function is monotone increasing, namely
if $F < G$ then $\frac{F}{F^2 + (1 - F )^2} < \frac{G}{G^2 + (1 - G )^2}$.
\end{lemma}
\begin{proof}
The proof is straightforward.
\end{proof}

\begin{lemma}
\label{lm:speed2}
If $F_0 < 1/{\sqrt{2}}$ then in at most $n+1$ iterations of the purification operation, where
\begin{equation}
\label{eq:speed2}
n = \left\lceil 
\frac{-\log_2 (F_0 \sqrt{2})}{\log_2 (F_0 + 1/2)}
\right\rceil 
\end{equation}
we have that $F_{n} \geq 1/{\sqrt{2}}$
\end{lemma}
\begin{proof}
Repeat purification as in $F_0 = F$ and $F_{n} = \frac{F_{n-1}^2}{F_{n-1}^2 + (1-F_{n-1})^2}$ and use Lemma~\ref{lm:speed1} to conclude that if $F_n < 1/{\sqrt{2}}$ then
\begin{align*}
\frac{F_{n}}{F_0} 
&= \frac{F_{n}}{F_{n-1}} \cdot \frac{F_{n-1}}{F_{n-2}} \cdots \frac{F_1}{F_0}
> \prod_{i=0}^{n-1} \left(F_i + \frac 12 \right) 
\geq \left(F_0 + \frac 12 \right)^n .
\end{align*}
We conclude that 
$$F_{n} \geq  \left(F_0 + \frac 12 \right)^n F_0 .
$$
It follows that the right-hand side above can be made greater than $1/{\sqrt{2}}$ for 
$
n \geq \frac{-\log_2 (F_0 \sqrt{2})}{\log_2 (F_0 + 1/2)} ,
$
which implies that
$F_{n} \geq 1/{\sqrt{2}}$. 
\end{proof}

We conclude by proving the following theorem. 

\begin{theorem}
\label{thm:speed1}
For any initial purification value $F_0$ such that $1/2 < F_0 < 1/{\sqrt{2}}$ and any $0 < \epsilon < 1$ arbitrarily small, in at most $m \leq \left\lceil \frac{-\log_2 (F_0 \sqrt{2})}{\log_2 (F_0 + 1/2)} \right\rceil  + \log_2 \log_2 (1/\epsilon)$ repetitions of the purification operation we will have that $F_m \geq 1 - \epsilon$. Moreover, if already $F_0 \geq 2/3$ then $m \leq  \log_2 \log_2 (1/\epsilon)$. 
\end{theorem}
\begin{proof}

From the discussion above, we see that starting from any initial value $F_0$ such that $1/2 < F_0 < 1/{\sqrt{2}}$ we can repeat the purification operation at $n$ times, where $n$ is given in Equation~\eqref{eq:speed2} so that $F_{n+1} \geq 1/{\sqrt{2}}$. 

Next, we indicate how many additional steps are needed to obtain accuracy. 
First of all, observe that $1/{\sqrt{2}} > 2/3$. A simple calculation using the definition of the purification operation shows that 
\begin{equation}
\label{eq:speed4}
\mbox{if $F= \frac{k}{k+1}$ then $F' = \frac{k^2}{k^2+1}$} .
\end{equation}
Because of Lemma~\ref{lm:speed0}, we may assume without loss of generality that $F_{n+1} = 2/3$.  Observe that $F_{n+1} = 2/3 = \frac{2^{2^0}}{2^{2^0} +1}$. By induction, if we assume that $F_{n+k}  = \frac{2^{2^{k-1}}}{2^{2^{k-1}}+1}$ then using the previously proved Assertion~\eqref{eq:speed4} we have that $F_{n+k+1}  = \frac{2^{2^{k}}}{2^{2^{k}}+1}$. 
It follows that for any required accuracy $\epsilon >0$, we have that
$
F_{n+k+1} > 1 - \epsilon ,
$
provided that $k > \log_2 \log_2 (1/\epsilon)$. 
\end{proof}

Next, Theorem~\ref{thm:repeated-fidelity-improvement} characterizes the memory cost of a one-time execution of purification and 
it provides an upper bound on the number of qubits required inside each repeater to support repeated execution of purification.

\begin{theorem}[Repeated Purification]
\label{thm:repeated-fidelity-improvement}
For a path of length $\ell$ (a positive integer), a repeater requires at most $2^{n + \ell - 1}$ qubits of memory to complete $n$ purification repetitions.
\end{theorem}
\begin{proof} 
The proof is by induction on the path's length $\ell$. 
Firstly, let us assume that $n$ is equal to one, 
i.e., one-time execution of purification.

\noindent
{\em Base Case: }($\ell = 1$).
The path consists of a source $s$ and a destination $d$, connected by a link.
The nodes $s$ and $d$ are terminals or repeaters.
Using two Bell pairs shared by $s$ and $d$, established with direct communications 
and purification, their fidelity is improved into a single Bell pair, as shown in Figure~\ref{fig:purificaton}. 
Every endpoint, $s$ and $d$, uses two qubits, which is equal to $2^{\ell}$ qubits, 
for $\ell = 1$.

\noindent
{\em Inductive Step:} ($\ell > 1$).
Let $\ell_1$ and $\ell_2$ be the number of links to the left and right of repeater $r$ on a path of length $\ell$ between nodes $s$ and $d$, 
with $\ell = \ell_1 + \ell_2$.
Clearly, both $\ell_1$ and $\ell_2$ are equal to or greater than one.
Let us assume that the nodes $s$ and $r$ (respectively, $r$ and $d$) share two Bell pairs. 
Using these four Bell pairs and two entanglement swapping operations, 
the repeater $r$ establishes two Bell pairs between $s$ and $d$. 
Purification improves their fidelity into a single Bell pair shared between $s$ and $d$. 
To create the two Bell pairs between $s$ and $r$ (respectively, $r, d$),
repeater $r$ requires $2 \cdot 2^{\ell_1}$ 
(respectively, $2 \cdot 2^{\ell_2}$) qubits,
for a total of 
$2 \cdot \left( 2^{\ell_1} + 2^{\ell_2} \right)$.
Because $\ell_1, \ell_2$ are greater than or equal to one but lower than $\ell$,
it is equal to or lower than $2 \cdot \left( 2^{\ell - 1} + 2^{\ell - 1} \right) = 2^{\ell}$ qubits.

This proves the theorem when $n$ is equal to one. 
Observe that the general statement of the theorem regarding the number $n$ of repetitions follows immediately by applying $n$ iterations of the above argument.
Since every additional repetition of the procedure by repeater $r$ multiplies the number of qubits by two, 
we have that $n$ repetitions need $2^{n-1} \cdot 2^{\ell} = 2^{n + \ell - 1}$ qubits.
\end{proof}

\begin{definition}[Purification procedure]
\label{def:purification-procedure}
Involving three nodes, where the middle one is a repeater $r$ while the two others $s$ and $d$ are terminals of repeaters, the purification procedure consists of seven operations, namely, the generation of four Bell pairs (two instances between $s$ to $r$ and two instances between $r$ to $d$), two entanglement swapping operations and one purification cycle under the control of $r$.
\end{definition}

\begin{corollary}
$n$ repetitions of purification by a repeater $r$ for a path of length $\ell$ requires at most $7 n ( \ell - 1)$ operations.
\end{corollary}
\begin{proof}
It follows from the definition of nested entanglement, the proof of Theorem~\ref{thm:repeated-fidelity-improvement}, as well as Definition~\ref{def:purification-procedure}.
\end{proof}

\subsection{Error Correction Analysis}

\label{sec:efficiency-analysis-EC}

As demonstrated in Lemma~\ref{lm:efficacy}, the efficacy of error correction can also be captured using the concept of fidelity.
Fidelity is used to quantifying the quality of Bell pairs established by a quantum system.
The fidelity of a Bell pair error-corrected with the $(3,1)$ repetition code is 
$\mathcal{F}(p)= (1 -p)^3 + 3p(1 - p)^2$, where $p$ is the probability of a qubit transformation from $\ket{0}$ to $\ket{1}$, or vice versa

\begin{definition}[$(3,1)$ repetition error correction concatenation for a single qubit]\label{def:repcodeconcatenation}
Building on Lemma~\ref{lm:efficacy}, 
let us define the fidelity of concatenated error corrections as
\begin{equation}
F_0=F(p) \mbox{ and }
F_{n} = \sqrt{ F_{n-1}^3 + 3(1 - F_{n-1})F_{n-1}^2 } \mbox{ for the number of concatenations } n=1,2,3\ldots.
\end{equation}
Moreover, let us define the sequence $\{ F_n : n\geq 1 \}$ with $F_0 = 1/2$.
\end{definition}

\begin{definition}[$(3,1)$ repetition error correction for a Bell pair]
Concatenated $n$ times, for the $(3,1)$ repetition error correction, let us define the fidelity of a Bell pair as
\begin{equation}
\label{eq:eq:error-correction-pair-fidelity}
\mathcal{F}_n = F_{n}^2 = F_{n-1}^3 + 3(1 - F_{n-1})F_{n-1}^2 .
\end{equation}
\end{definition}

\begin{lemma}
\label{lm:increasing}
The sequence $\{ F_n : n\geq 1 \}$ is monotone non-decreasing, and its limit as $n$ goes to infinity equals one.
\end{lemma}
\begin{proof}
Consider the function $f(x)=\sqrt{x^3 + 3(1 - x)x^2}$. Clearly, $f(x) = x\sqrt{3-2x}$. 
It is straightforward to verify that $f(x) \leq 1$, for all $0 \leq x \leq 1$. 
By definition of $F_n$, we have that
$$
F_n = F_{n-1} \sqrt{F_{n-1} + 3(1-F_{n-1})} = F_{n-1} \sqrt{3-2F_{n-1}},
$$
for $n \geq 1$. However, $3-2x \geq 1$, for $0 \leq x \leq 1$. 
Using this and the fact that $F(p) = p\sqrt{3-2p}$ we see that $F(p) \geq p$, 
for all $0 \leq p \leq 1$.
Therefore $F_n \geq F_{n-1}$, for all $n\geq 2$. 
It follows that the sequence $\{ F_n : n \geq 0 \}$ is monotone non-decreasing. 
Consequently, its limit $f := \lim_{n \to + \infty} F_n$ exists.
It is now shown that the limit $f$ must equal one. 
Indeed, consider the defining equation of $F_n$, namely $F_n = F_{n-1} \sqrt{3-2F_{n-1}}$. 
In this equation, passing to the limit as $n \to + \infty$ we see that $f = f \sqrt{3 -2f}$. 
The only solution to this equation is $f=1$.
\end{proof}

Now, let us estimate the convergence speed of the sequence $F_n$. First, consider the formula $F_n = F_{n-1} \sqrt{3-2F_{n-1}}$. 
Let $F_{n-1} = 1 -\epsilon_{n-1}$, for some $\epsilon_{n-1} > 0$. Observe that 
$$
F_n = (1 -\epsilon_{n-1}) \sqrt{3 -2 (1 - \epsilon_{n-1})} = (1-\epsilon_{n-1}) \sqrt{1+ 2\epsilon_{n-1} } .
$$
It can be verified that the inequality
$$
(1-\epsilon_{n-1}) \sqrt{1+ 2\epsilon_{n-1} } > 1 - \epsilon_{n-1} /2
$$ 
is valid as long as $\epsilon_{n-1} < \frac{6 - \sqrt{20}}{8} \approx 0.191 \ldots$ (to see this claim, multiply out and use the resulting quadratic in the variable $\epsilon_{n-1}$). 
Therefore,  we can conclude that $F_n = 1 - \epsilon_n$, where $\epsilon_n \leq \epsilon_{n-1}/ 2$. 
As depicted in Figure~\ref{fig:repetvsconcat}, when starting with fidelity $F_0 > 1/2$, 
after four iterations we can attain $\epsilon_{n_0} \leq 0.1 < \frac{6 - \sqrt{20}}{8}$. 
Therefore, convergence of the fidelity $F_{n}$ to one is exponentially fast, namely $\epsilon_{n} \leq \epsilon_{n_0}/ 2^{n-n_0}$, where $n > n_0$. 
To conclude, we have the following theorem.
\begin{theorem}
If $1/2 < F_0 \leq 1$ then $(3,1)$ repetition error concatenation fidelity $F_n \geq 1 - \epsilon$ can be reached in at most $\left\lceil\log_2 (c /\epsilon)\right\rceil$ steps, where $c>0$ is constant. 
\end{theorem}

\begin{theorem}[Concatenated error correction]
\label{thm:concatenated-error-correction}
Assume that a one-time execution of error correction requires $m$ ancillary qubits at every subpath's endpoint.
For a path of length $\ell$, a repeater requires at most $m^{n + \ell - 1}$ qubits of memory to complete $n$ error correction concatenations, $\ell$, $m$ and $n$ are positive integers, $\ell, n \ge 1$ and $m \ge 3$.
\end{theorem}
\begin{proof} 
The proof is by induction on the path's length $\ell$.
Firstly, let us assume that $n$ is equal to one, 
i.e., one-time execution of error correction.

\noindent
{\em Base Case: }($\ell = 1$).
The path consists of a source $s$ and a destination $d$, connected by a link.
The nodes $s$ and $d$ are terminals or repeaters.
Using direct communications,
a $m$ qubit block possessed by $s$ entangled with a $m$ qubit block possessed by $d$ are established, as depicted in Figure~\ref{fig:error-correction-circuit-3}. 
Using error correction, they are decoded into a single qubit pair shared between $s$ and $d$. 
The endpoints $s$ and $d$ require $m$ qubits, 
which is equal to $m^{\ell}$, for $\ell = 1$ and $m \ge 3$.

\noindent
{\em Inductive Step:} ($\ell > 1$).
Let $\ell_1$ and $\ell_2$ be the number of links to the left and right of repeater $r$ on a path of length $\ell$ between nodes $s$ and $d$, 
with $\ell = \ell_1 + \ell_2$.
Both $\ell_1$ and $\ell_2$ are greater than or equal to one.
The nodes $s$ and $d$ can be either terminals or repeaters.
Let us assume the nodes $s$ and $r$ (respectively, $r$ and $d$) share a Bell pair. 
As depicted in Figure~\ref{fig:swapping-error-correction-circuit},
using entanglement swapping
the repeater $r$ establishes entanglement between a $m$ qubit block possessed by $s$ and a $m$ qubit block possessed by $d$. 
Using error correction, they are decoded  into a single qubit pair shared between $s$ and $d$. 
The repeater $r$ uses one qubit in each Bell pair, i.e., two qubits.
Every node $s$ and $d$ used $m$ ancillary qubits.
To create the Bell pair between $s$ and $r$ (respectively, $r$ and $d$) it requires at most $m^{\ell_1}$ (respectively, $m^{\ell_2}$) qubits,
for a total of 
$m^{\ell_1} + m^{\ell_2}$ quibits.
Because $\ell_1, \ell_2$ are greater than or equal to one but lower than $\ell$,
it is lower than or equal to $m^{\ell-1} + m^{\ell-1}$.
Which is lower than or equal to $m^{\ell}$ qubits.

This proves the theorem when $n$ is equal to one. 
Observe that the general statement of the theorem regarding the number $n$ of concatenations (as depicted in Figure~\ref{fig:swapping-error-correction-repetition}) follows immediately by applying $n$ iterations of the above argument.
Since every additional concatenation of the procedure by repeater $r$ multiplies the number of qubits by $m$, we have that
$k$ concatenations would need $m^{n-1} \cdot m^{\ell}$, or
$m^{n + \ell - 1}$ qubits.
\end{proof}

\begin{definition}[Error correction procedure]
\label{def:correction-procedure}
For a $(3,1)$ error correction code,
involving three nodes, where the middle one is a repeater $r$ while the two others $s$ and $d$ are terminals of repeaters, the error correction procedure consists of four operations, namely, the generation of two Bell pairs (between $s$ to $r$ and between $r$ to $d$), one entanglement swapping and one error correction operation performed by $s$ and $d$, at their location.
\end{definition}

\begin{corollary}
\label{cor:ECCOps}
Concatenation of error correction by a repeater for a path of length $\ell$ requires at most $4(\ell - 1)$ operations. 
\end{corollary}

Notice, from Corollary~\ref{cor:ECCOps}, that concatenated error correction, compared to repeated purification, only introduces a constant number of  ancillary qubits per node, independently from the number of iterations, i.e., the number of operations is proportional to the path length but independent of the number of concatenations.

\section{Analytical and Simulation Results}\label{sec:generalgraphs}

\begin{figure*}[!b]
	\subfigure[Repeated purification.]{
	    \includegraphics[width=6.25cm]{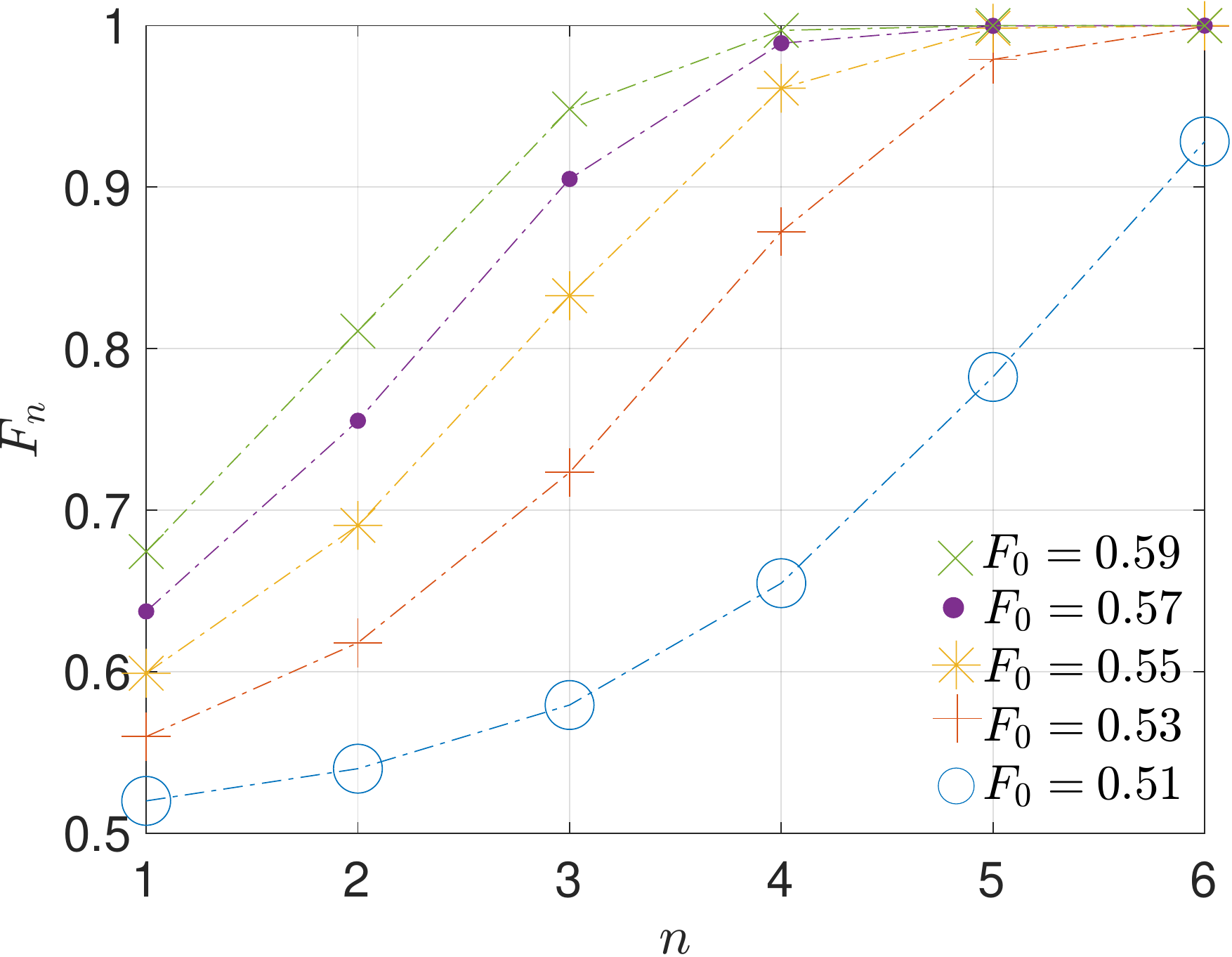}
	}
	\subfigure[Concatenated error correction, $(3,1)$ repetition code.]{
	    \includegraphics[width=6.25cm]{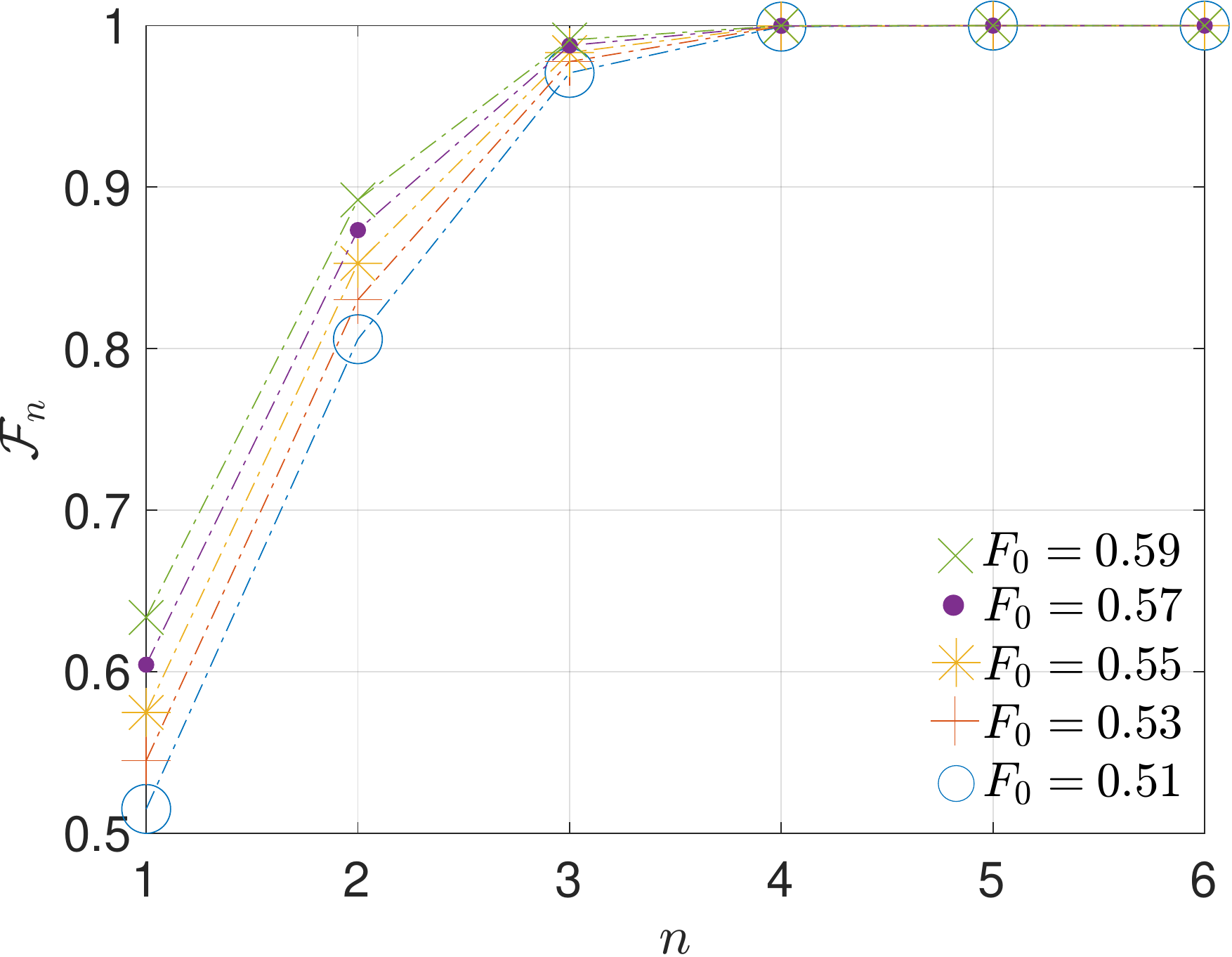}
	}
	\caption{Fidelity ($F_n$ or $\mathcal{F}_n$) vs. number of repetitions or concatenations ($n$), for initial fidelity values $0.51, 0.53, 0.55, 0.57$, and $0.59$. We can observe that $\mathcal{F}_n$ tends to value one faster than $F_n$, for all initial fidelity values $F_0$.}
	\label{fig:repetvsconcat}
\end{figure*}

Figure~\ref{fig:repetvsconcat}~(a) shows the evolution of the fidelity sequence $F_n$ versus the number of repeated purifications $n$, 
for initial fidelity $F_0$ values $0.51, 0.53, 0.55, 0.57$ and $0.59$ (Equation~\ref{eq:purent4}). 
The plot shows that $F_n$ rapidly tends to value one, 
with initial fidelity greater than $0.5$.
In Figure~\ref{fig:repetvsconcat}~(b), we have the same type of graph but for concatenated error correction, with the $(3,1)$ repetition code (Equation~\ref{eq:eq:error-correction-pair-fidelity}).
$\mathcal{F}_n$ tends to value one faster than $F_n$.

\begin{figure*}[!t]
	\subfigure[Repeated purification.]{
    \includegraphics[width=6.25cm]{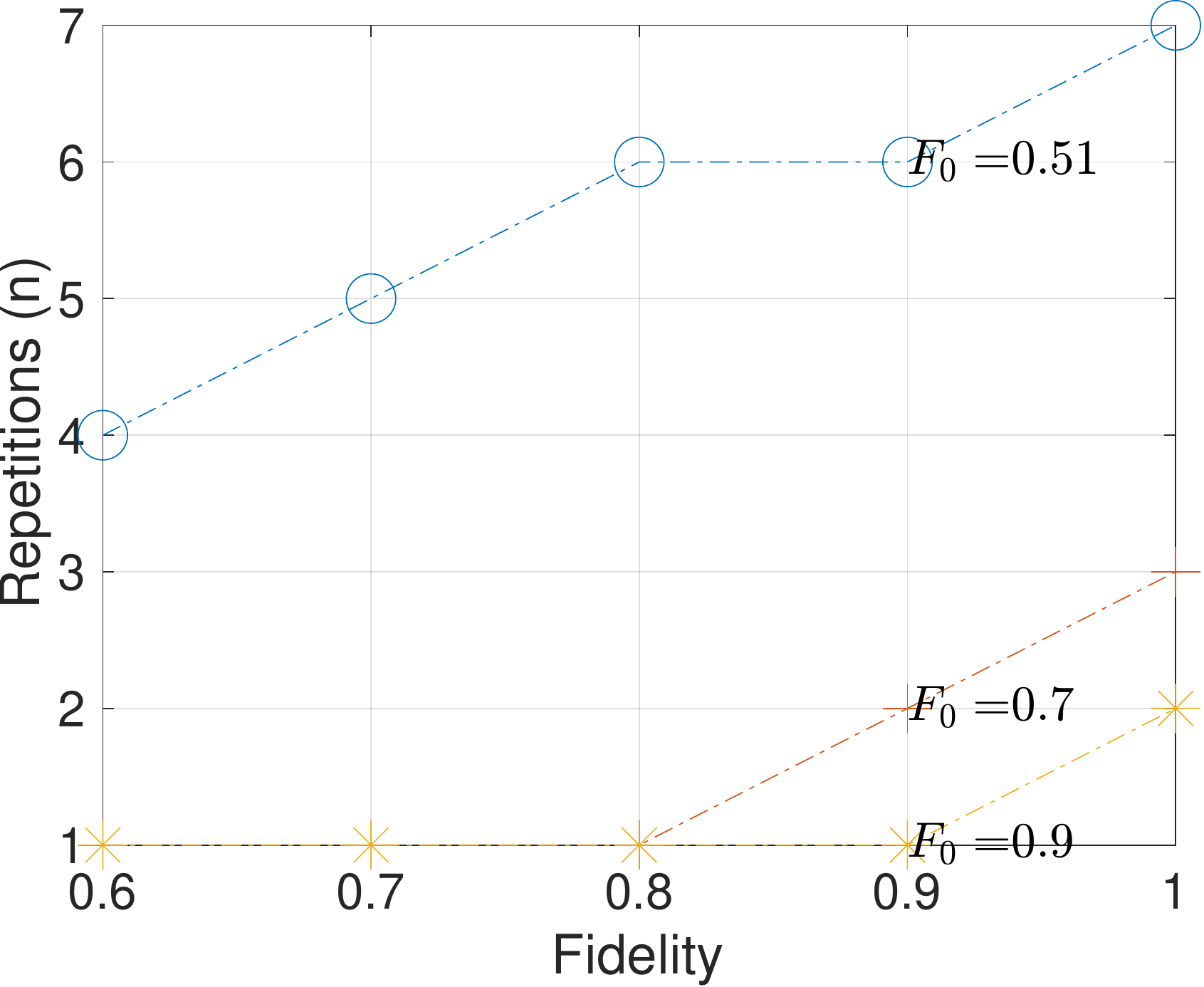}
    }
    \subfigure[Concatenated error correction.]{
    \includegraphics[width=6.25cm]{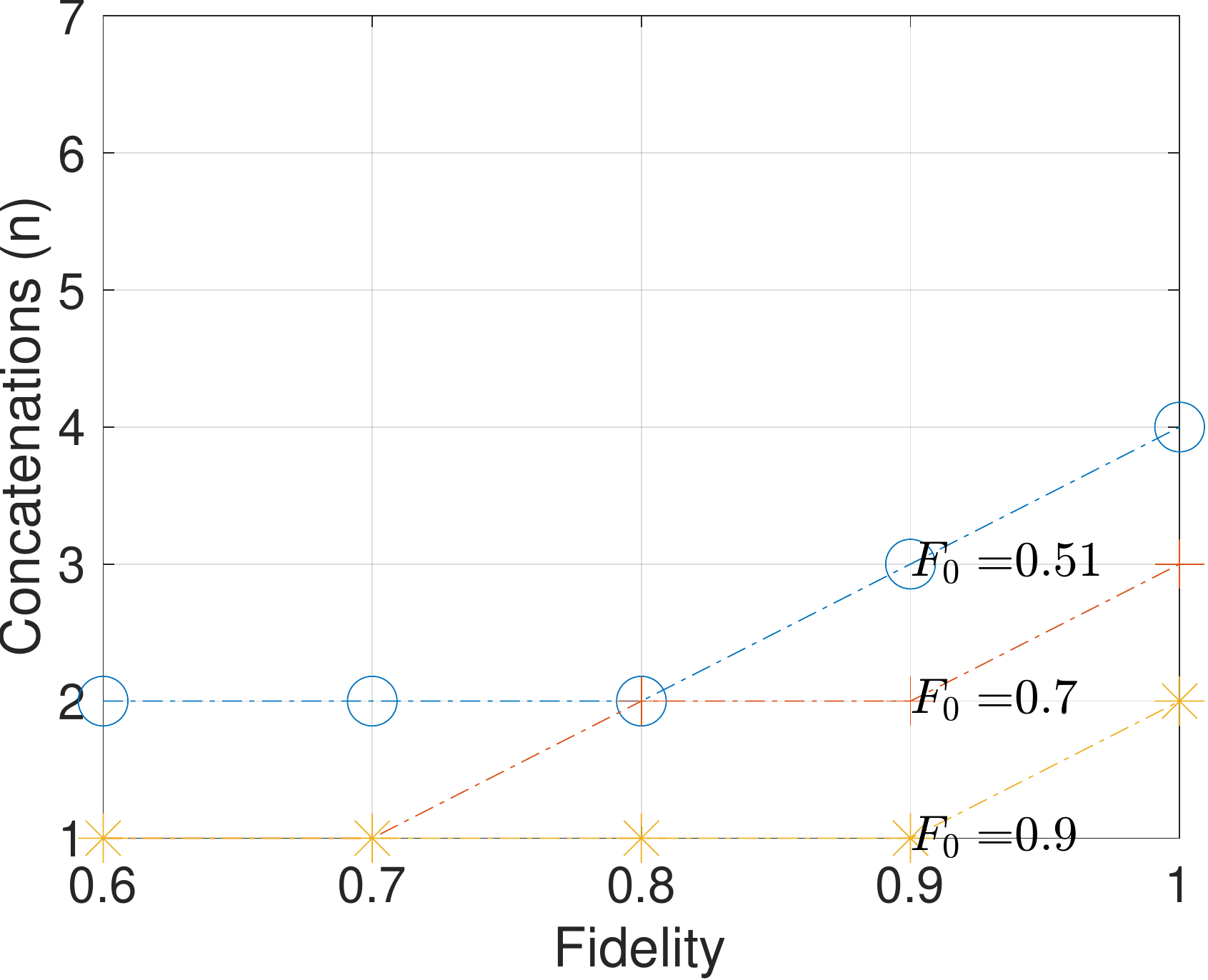}
    }
    \subfigure[Repeated purification.]{
    \includegraphics[width=6.25cm]{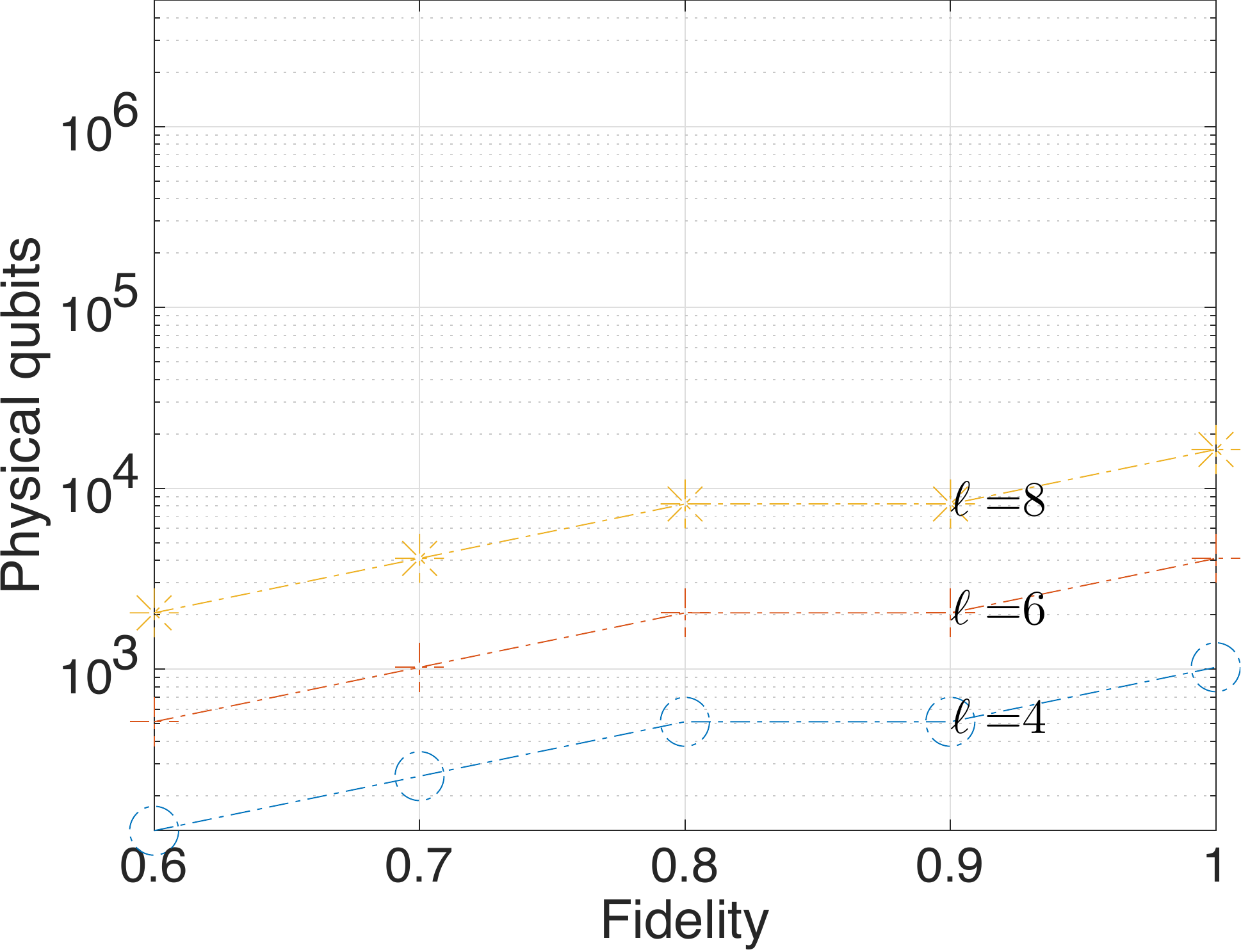}
    }
    \subfigure[Concatenated error correction.]{
    \includegraphics[width=6.25cm]{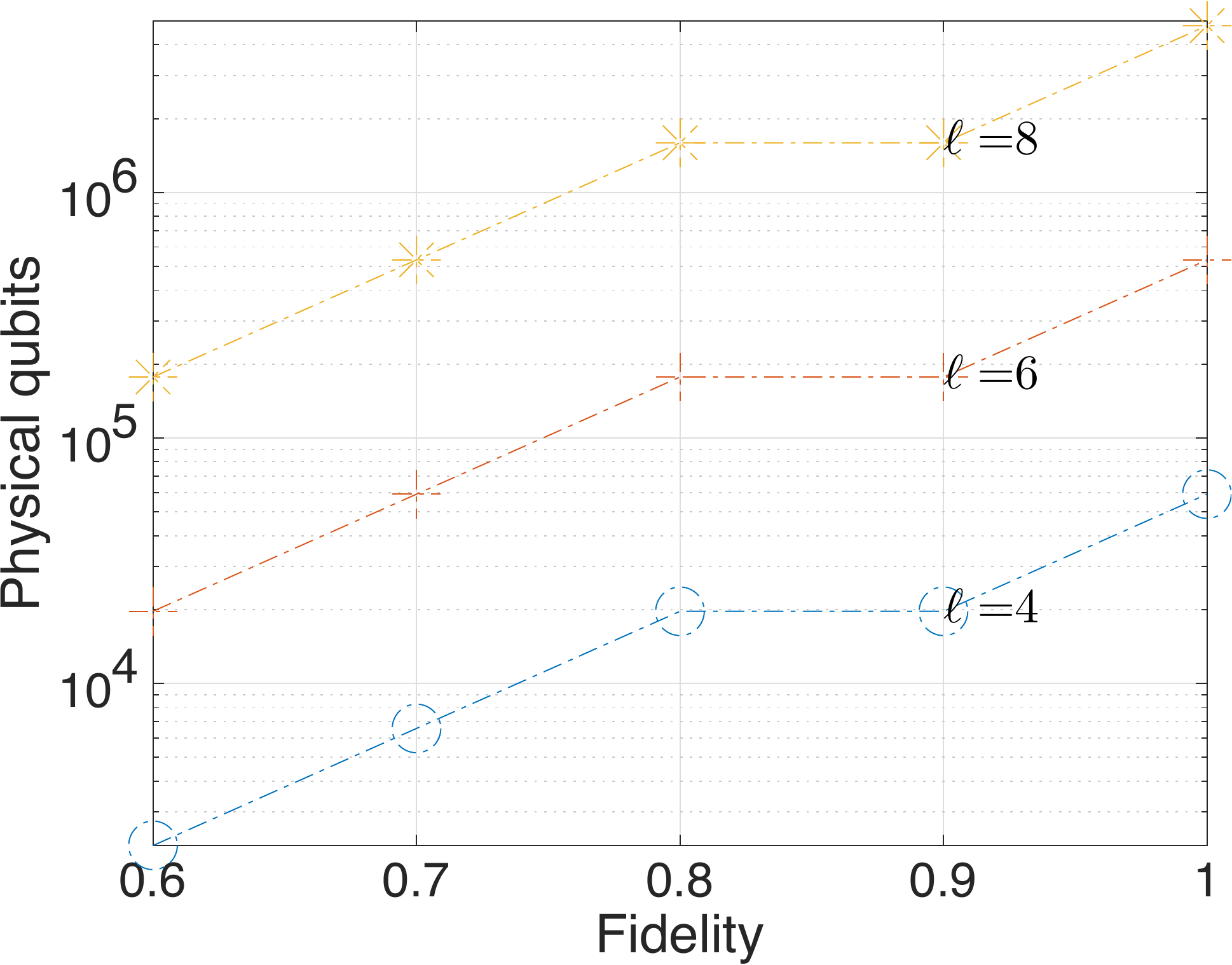}
    }
    \subfigure[Repeated purification.]{
    \includegraphics[width=6.25cm]{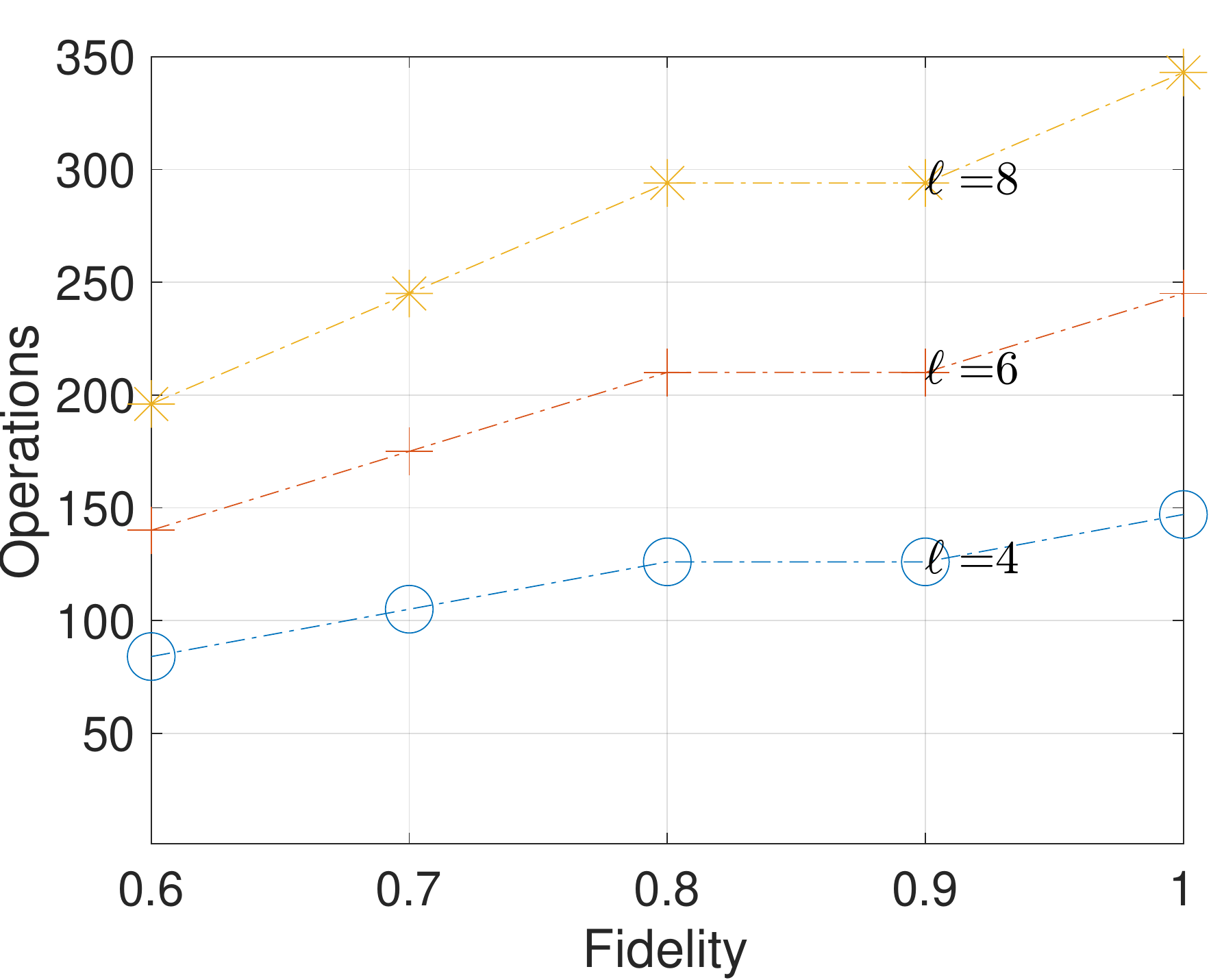}
    }
    \subfigure[Concatenated error correction.]{
    \includegraphics[width=6.25cm]{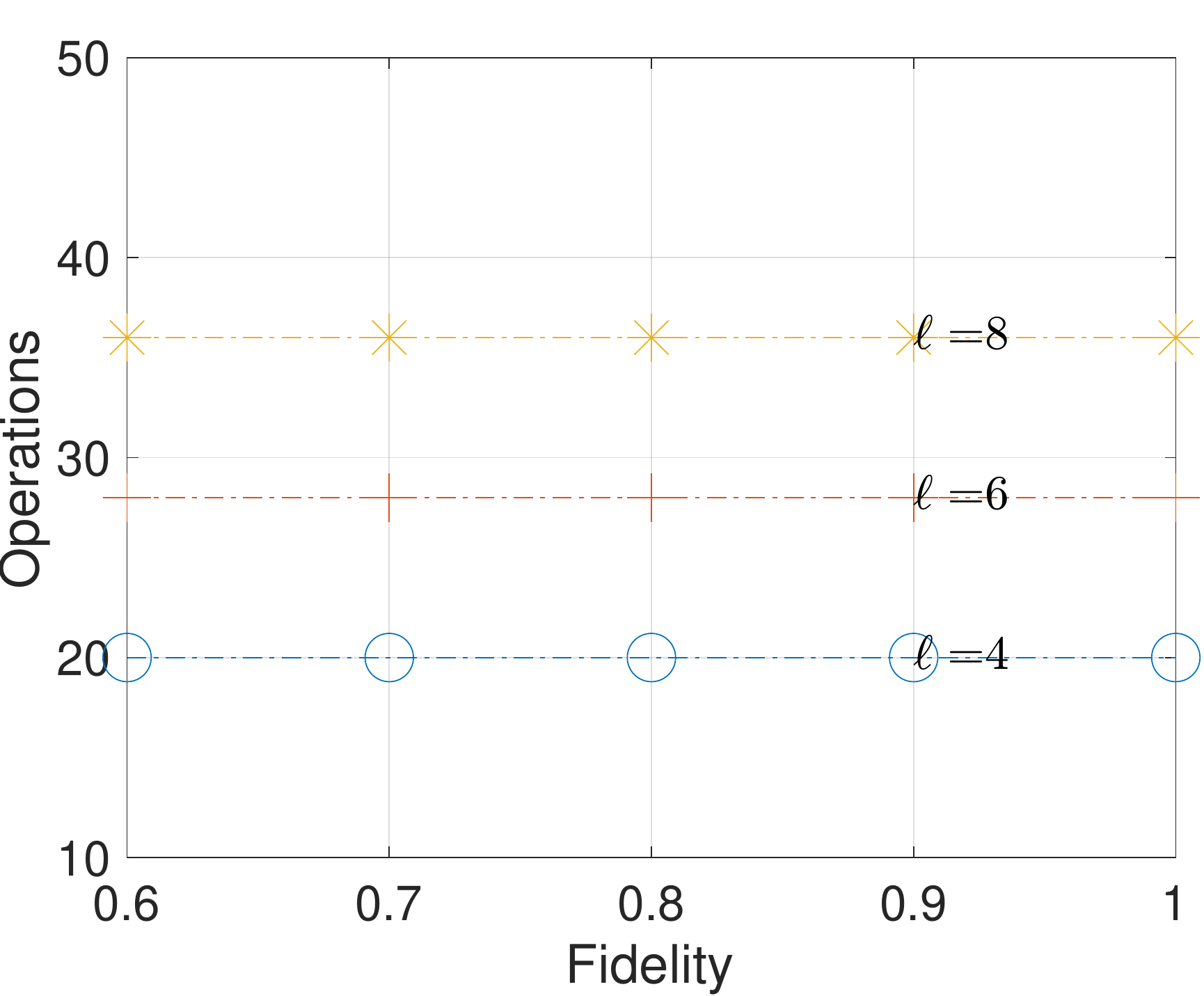}
    }
    \caption{(a,b) Number of repetitions and concatenations ($n$) versus fidelity ($F_n$ or $\mathcal{F}_n$). (c,d) Number of required qubits versus achieved fidelity ($F_n$ or $\mathcal{F}_n$). (e,f) Number of required operations versus achieved fidelity ($F_n$ or $\mathcal{F}_n$).}
	\label{fig:nvsFnrepetvsconcat}
\end{figure*}

Figures~\ref{fig:nvsFnrepetvsconcat}~(a,b) plot the value of $n$, i.e., the number of repetitions and concatenations required to achieve a given fidelity, $F_n$ or $\mathcal{F}_n$.
The initial fidelity ($F_0$) is $0.51, 0.75$, and $0.9$.
There are data points for repeated purification (a) and concatenated error correction concatenation for the $(3,1)$ repetition code (b).
In the $0.51$ case, with low initial fidelity,
error correction requires fewer concatenations than purification repetitions to achieve a given fidelity.
They are almost the same in the $0.7$ case.
They are the same in the $0.9$ case. Figures~\ref{fig:nvsFnrepetvsconcat}~(c,d) plot the numbers of qubits needed to reach a given fidelity ($F_n$ or $\mathcal{F}_n$).
The initial fidelity ($F_0$) is $0.51$.
The path length ($\ell$) is 4, 6 or 8~hops.
Repeated purification achieves near 100\% fidelity using fewer physical qubits than concatenated error correction. Figures~\ref{fig:nvsFnrepetvsconcat}~(e,f) plot the numbers of operations needed to reach a given fidelity ($F_n$ or $\mathcal{F}_n$).
The initial fidelity ($F_0$) is $0.51$.
The path length ($\ell$) is 4, 6 or 8~hops.
Error correction achieves near 100\% fidelity using fewer operations than repeated purification.

\begin{figure}[!hptb]
\begin{center}
\includegraphics[width=10cm]{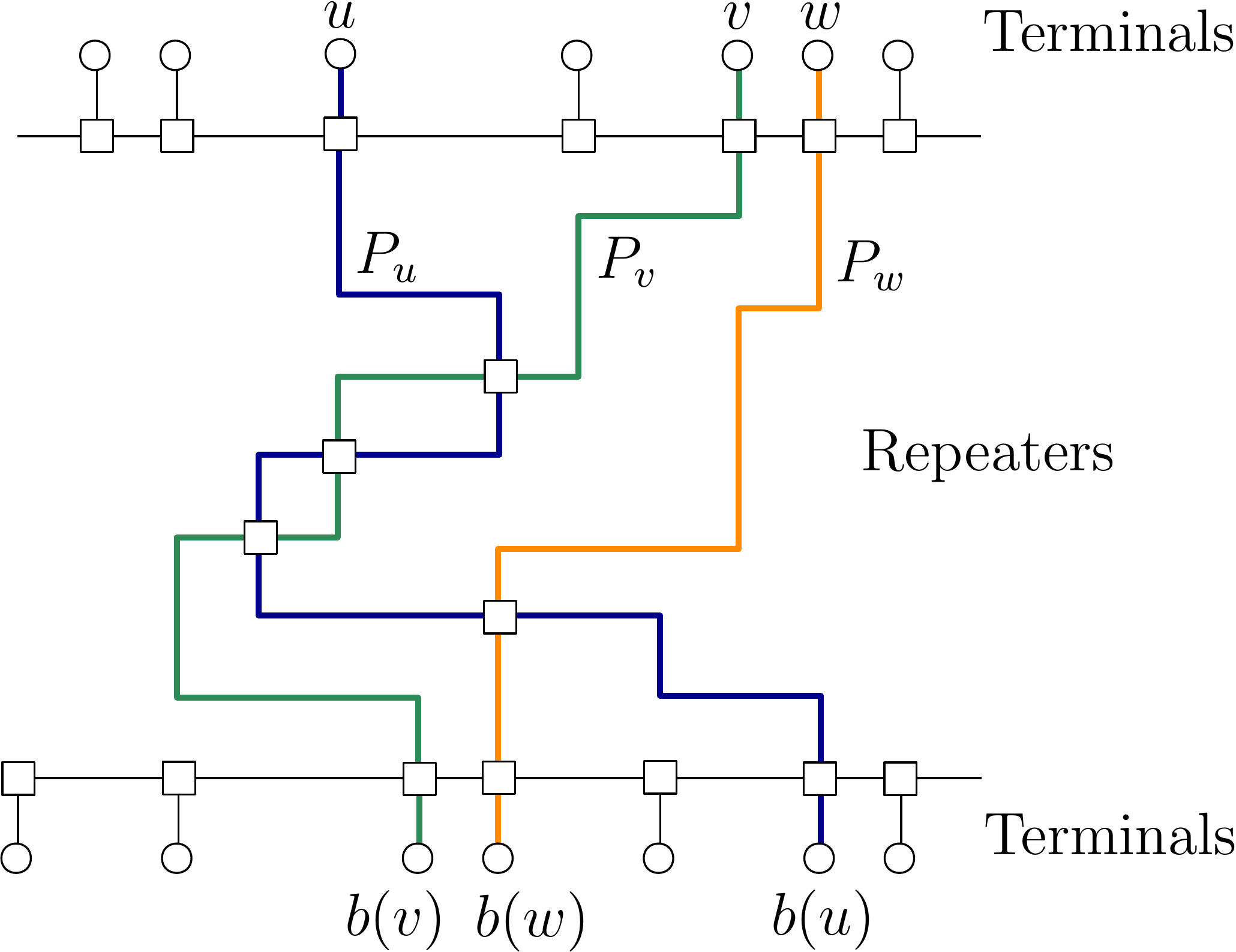}
\caption{A sparse grid topology (not depicted) with terminals at the top and bottom and repeaters in between. Repeaters are depicted with $\Box$ and terminals with {\scriptsize $\bigcirc$}. Terminals at the top and bottom are numbered $1,2,\ldots ,n$, respectively. A possible arrangement of three paths $P_u, P_v, P_w$ of active terminals $u<v<w$ in a grid topology of repeaters such that $b(v) \geq b(w) \geq b(u)$.}
\label{default2}
\end{center}
\end{figure}

Let us now consider a $k$ by $k$ sparse grid topology with $k$ terminals at the top and $k$ terminals at the bottom rows (see~Figure~\ref{default2}). The goal is to determine the memory requirements and operational complexity of the most congested repeaters in the network. We do not assume full connectivity among terminals. 
Instead, with probability $p$, a terminal $u$ at the top of the grid becomes active and intends to communicate to some terminal at the bottom row of the grid. Therefore in expectation, $pk$ such terminals become active. Each terminal $u$ at the top selects a terminal $b(u)$ among the terminals in the bottom row. For each terminal $u$ at the top, the terminal $b(u)$ is chosen randomly among the $k$ terminals at the bottom. Moreover, the choice of any two terminals $u,v$ are independent of each other. For each terminal $u$, consider a path $P_u$ of repeaters connecting terminals $u$ and $b(u)$.  The path $P_u$ may be chosen by any standard procedure considering the grid topology of active repeaters using Dijkstra's algorithm.
It does not have to have optimal length. 
We are interested in the expected number of crossings  (see~Figure~\ref{default2}).  
The actual number of crossings may well depend on the paths $\{ P_u: 1 \leq u \leq k \}$ selected. 

The following argument gives a lower bound on the expected number of crossings. By the previous discussion, the expected number of  active terminals at the top row is $pk$. We say that the order of a pair of terminal $\{ u, v\}$ at the top is reversed if $u<v$ and $b(u) \geq b(v)$. It is clear that, regardless of the choice of paths $P_u, P_v$, there is a crossing between them provided there is a reversal, i.e., $u<v$ and $b(u) \geq b(v)$. 
Note that at every crossing, between $P_u, P_v$, at least one repeater (possibly more) must serve both paths. Observe that if $u < v$ then we have that 
\begin{align*}
\Pr [P_u \mbox{ crosses } P_v ] 
&\geq \sum_{i \geq j} \Pr [b(u) = i ~\&~ b(v) = j]\\
&= \sum_{i \geq j} \Pr [b(u) = i] \cdot \Pr [ b(v) = i] \\
&= \sum_{i \geq j} \frac 1{k^2} = \frac{k^2+k}{2k^2} = \frac 12 + \frac 1{2k}.
\end{align*}

\begin{figure}[!b]
\centering
\subfigure[Repeated purification]{
\includegraphics[width=6cm]{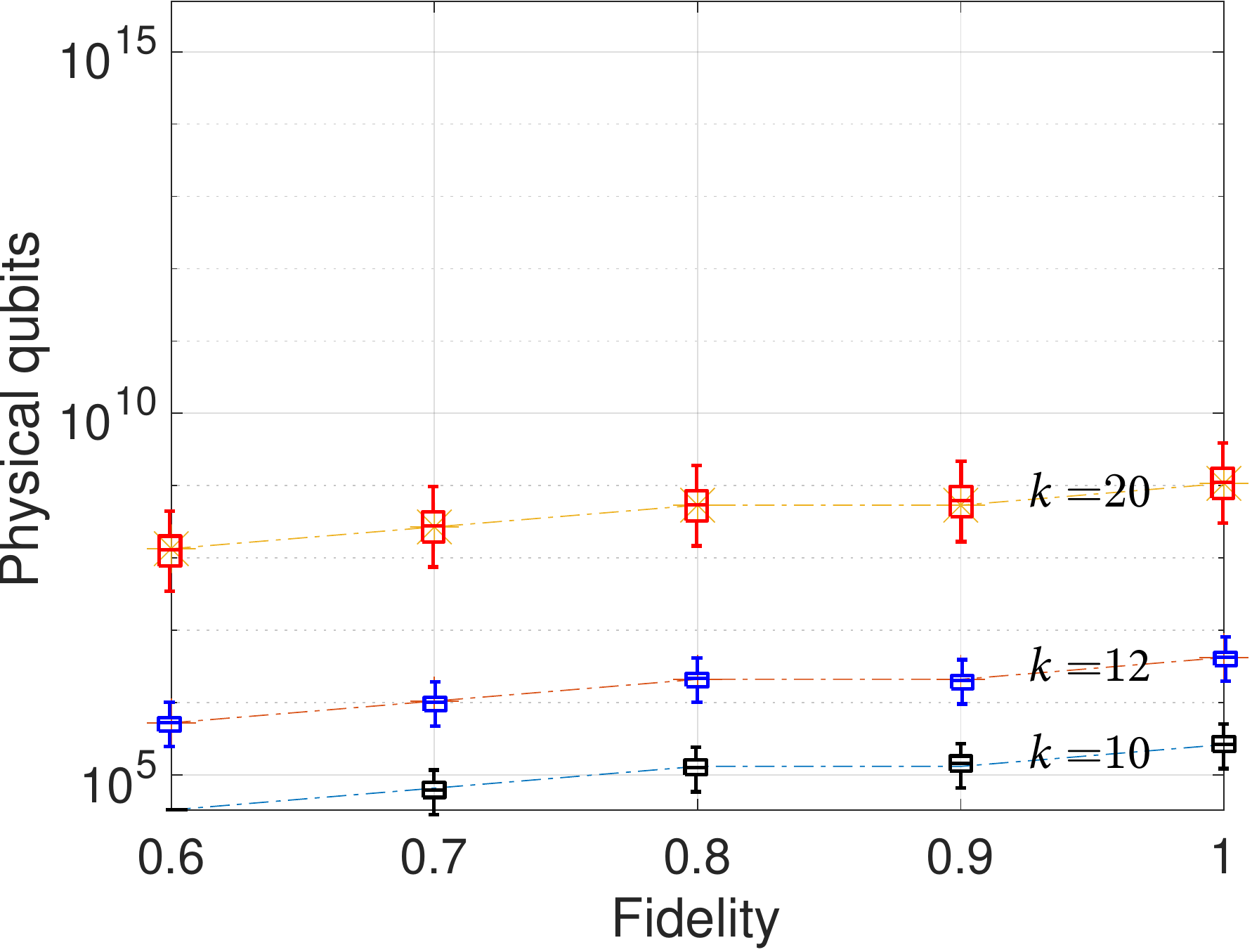}
}
\subfigure[Concatenated error correction]{
\includegraphics[width=6cm]{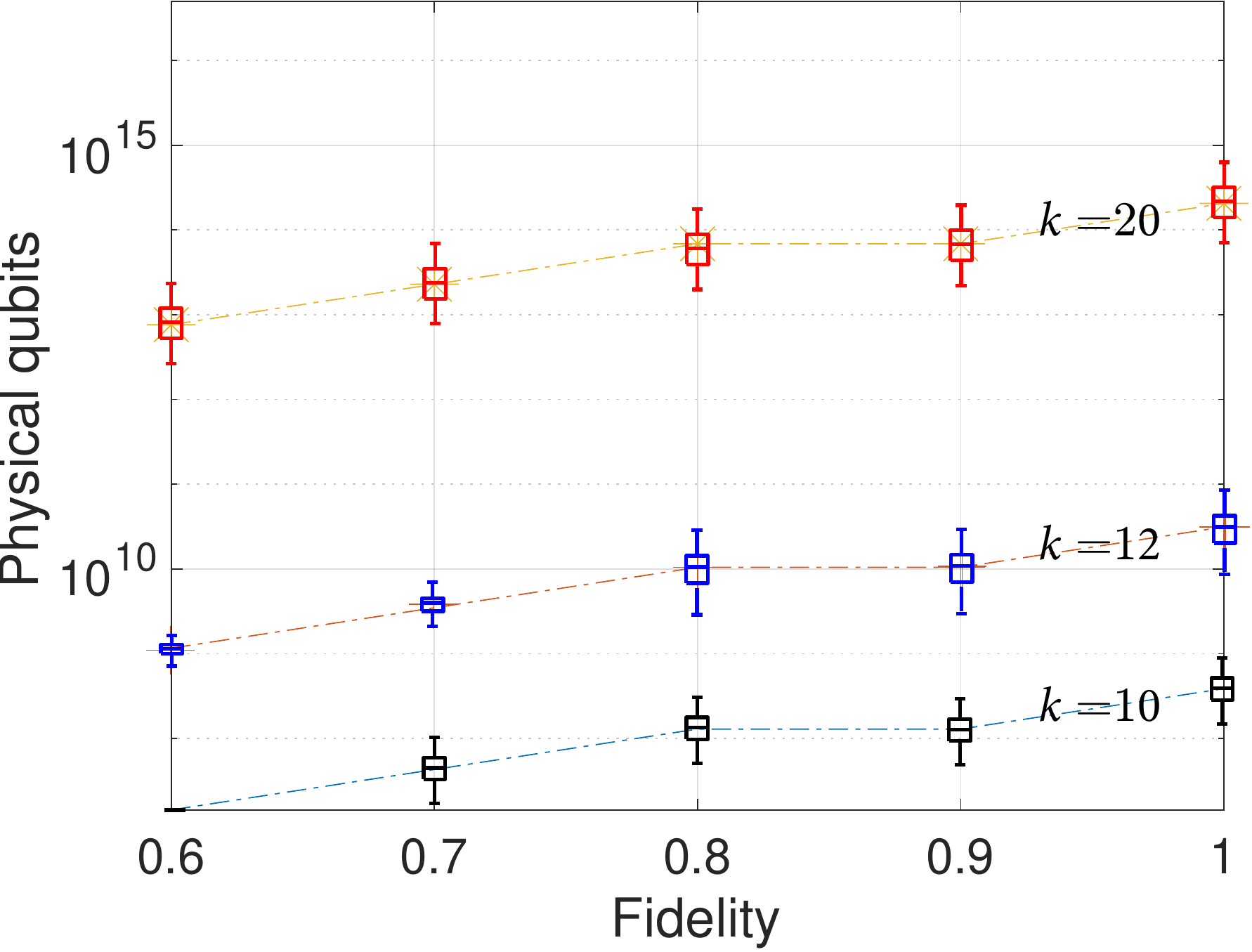}
}
\subfigure[Repeated purification]{
\includegraphics[width=6cm]{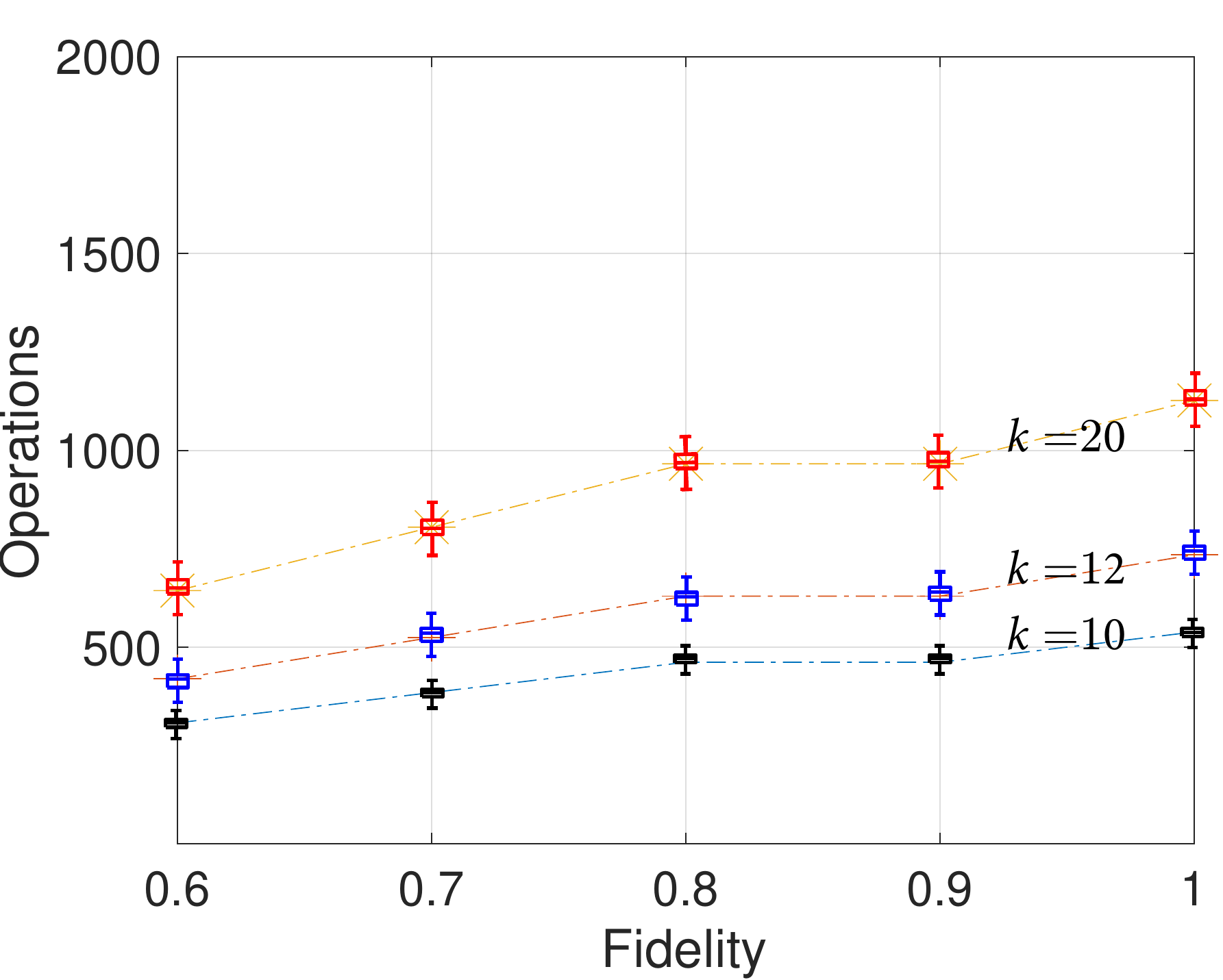}
}
\subfigure[Concatenated error correction]{
\includegraphics[width=6cm]{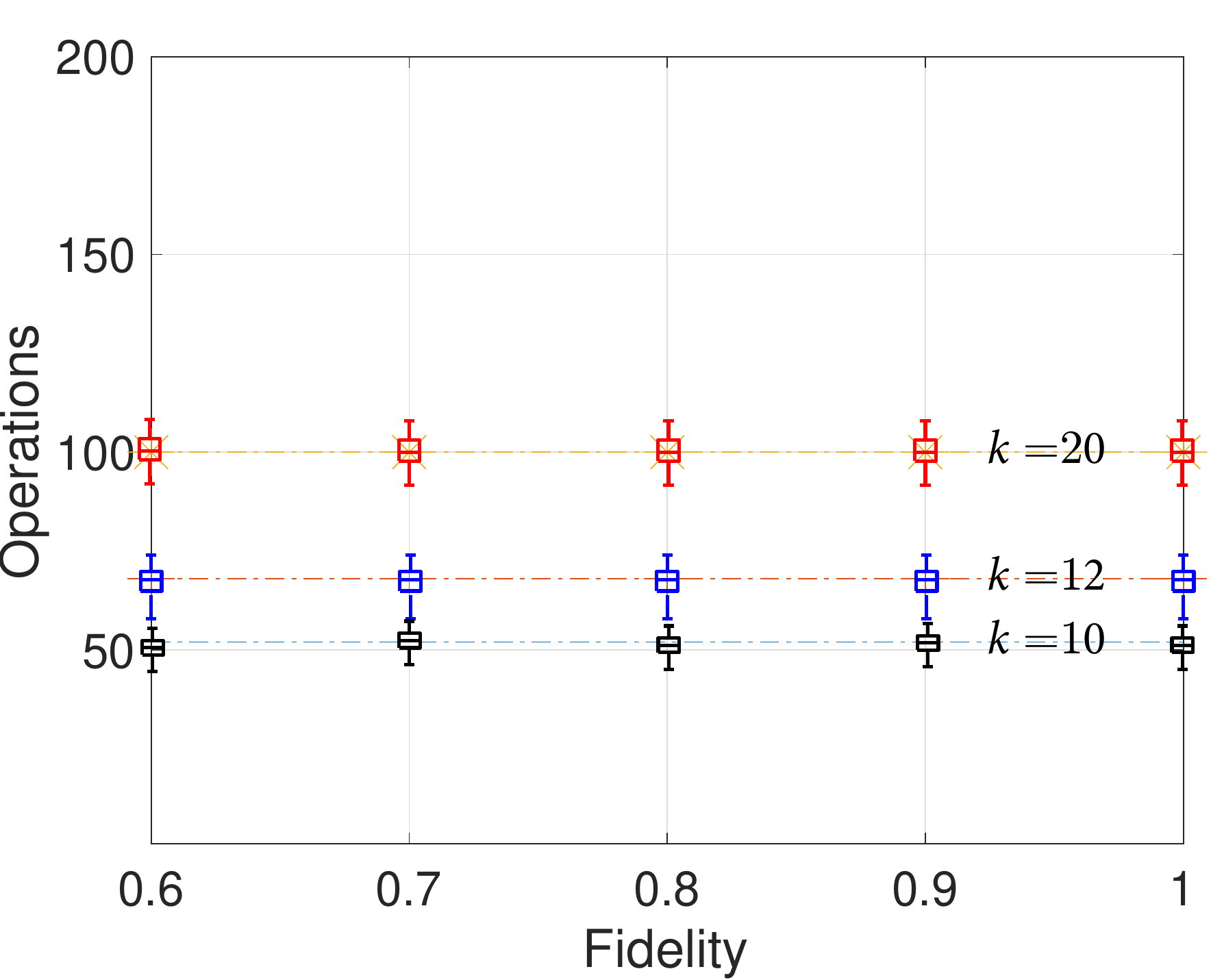}
}
\caption{Extended results using the sparse grid $k$ by $k$ scenario, in which $10 \leq k \leq 20$. (a,b) Concatenated error correction consumes more physical qubits. (c,d) Repeated purification needs more operations.\label{fig:pythonResults2}}
\end{figure}

If we assume that the random variable $I_{uv}$ indicates that the path $P_u$ crosses the path $P_v$, 
then it follows that the expectation of the random variable ${\mathcal C}$ which counts the total number of crossings must satisfy
\begin{equation}
\label{eq:cross}
E[{\mathcal C}] \geq (pk)^2  \left( \frac 12 + \frac 1{2k} \right) .
\end{equation}

\begin{figure}[!b]
\centering
\subfigure[Repeated purification]{
\includegraphics[width=6cm]{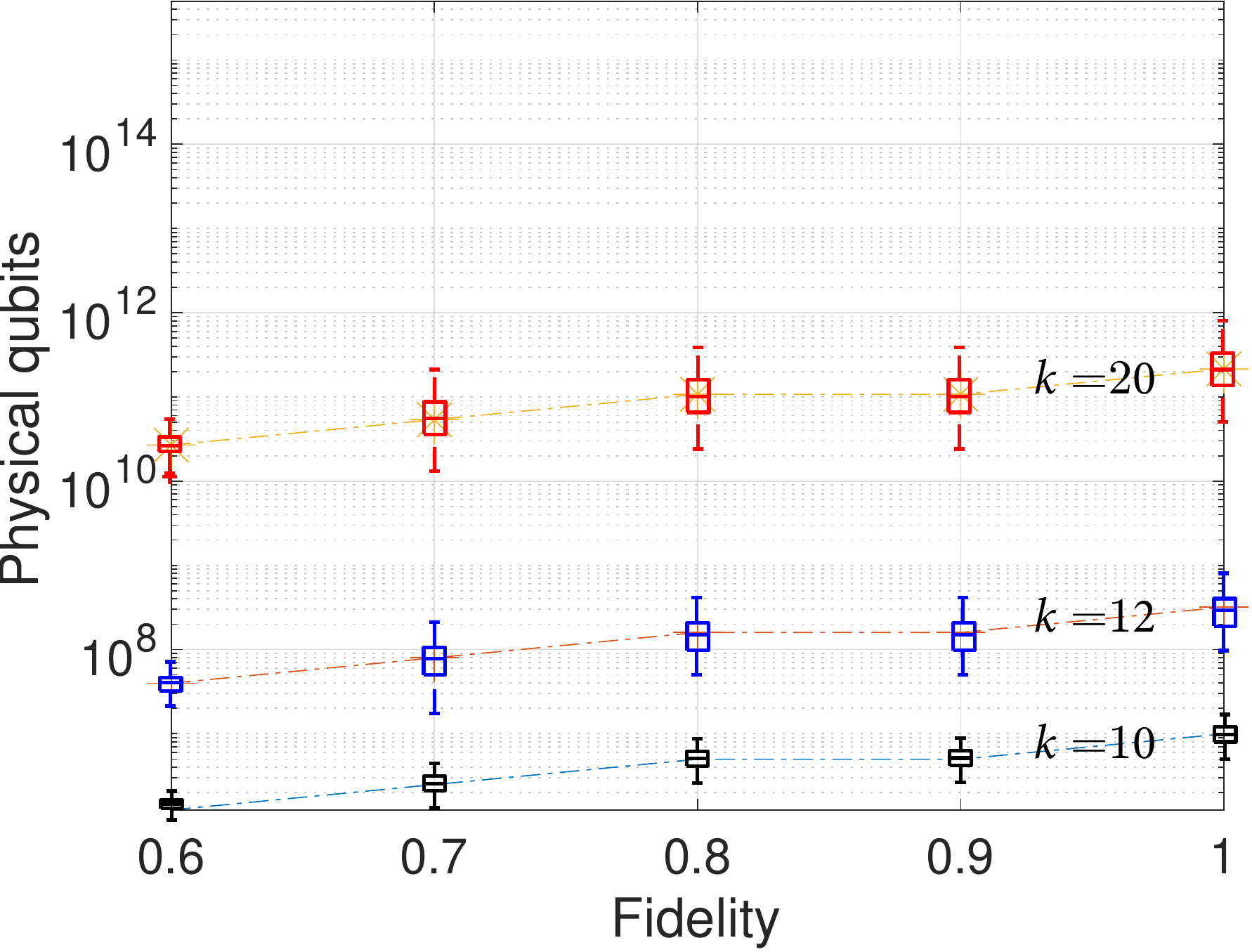}
}
\subfigure[Concatenated error correction]{
\includegraphics[width=6cm]{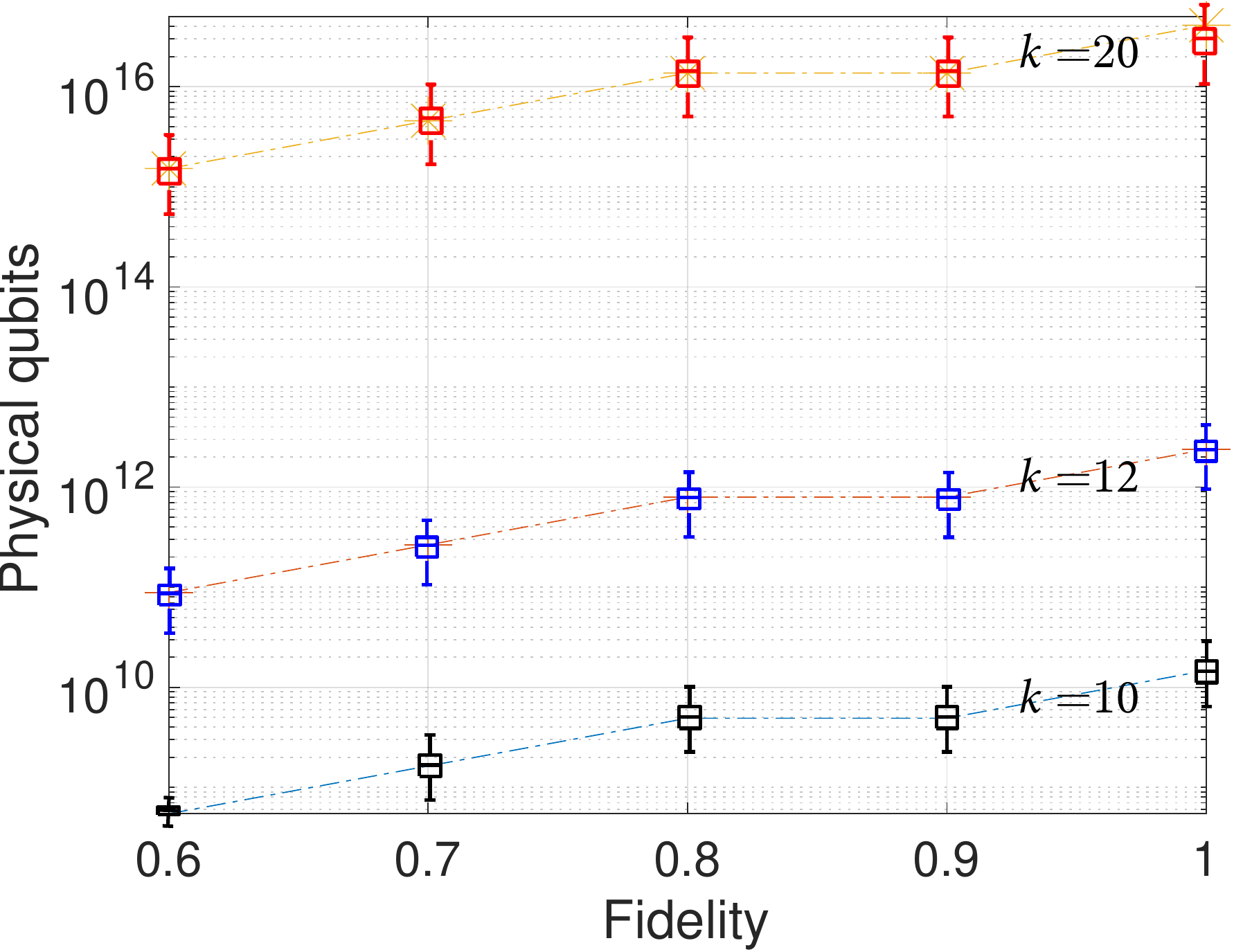}
}
\subfigure[Repeated purification]{
\includegraphics[width=6cm]{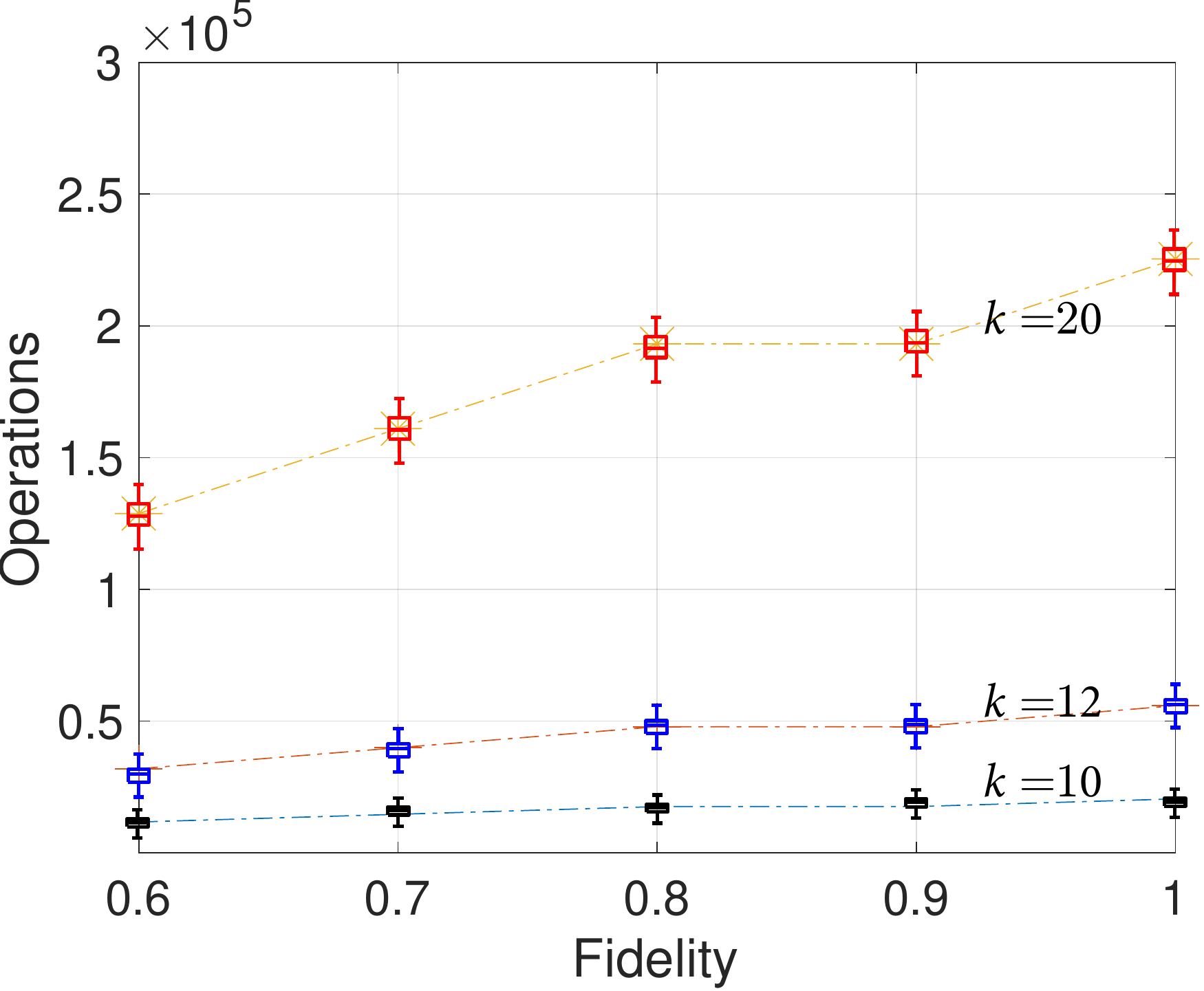}
}
\subfigure[Concatenated error correction]{
\includegraphics[width=6cm]{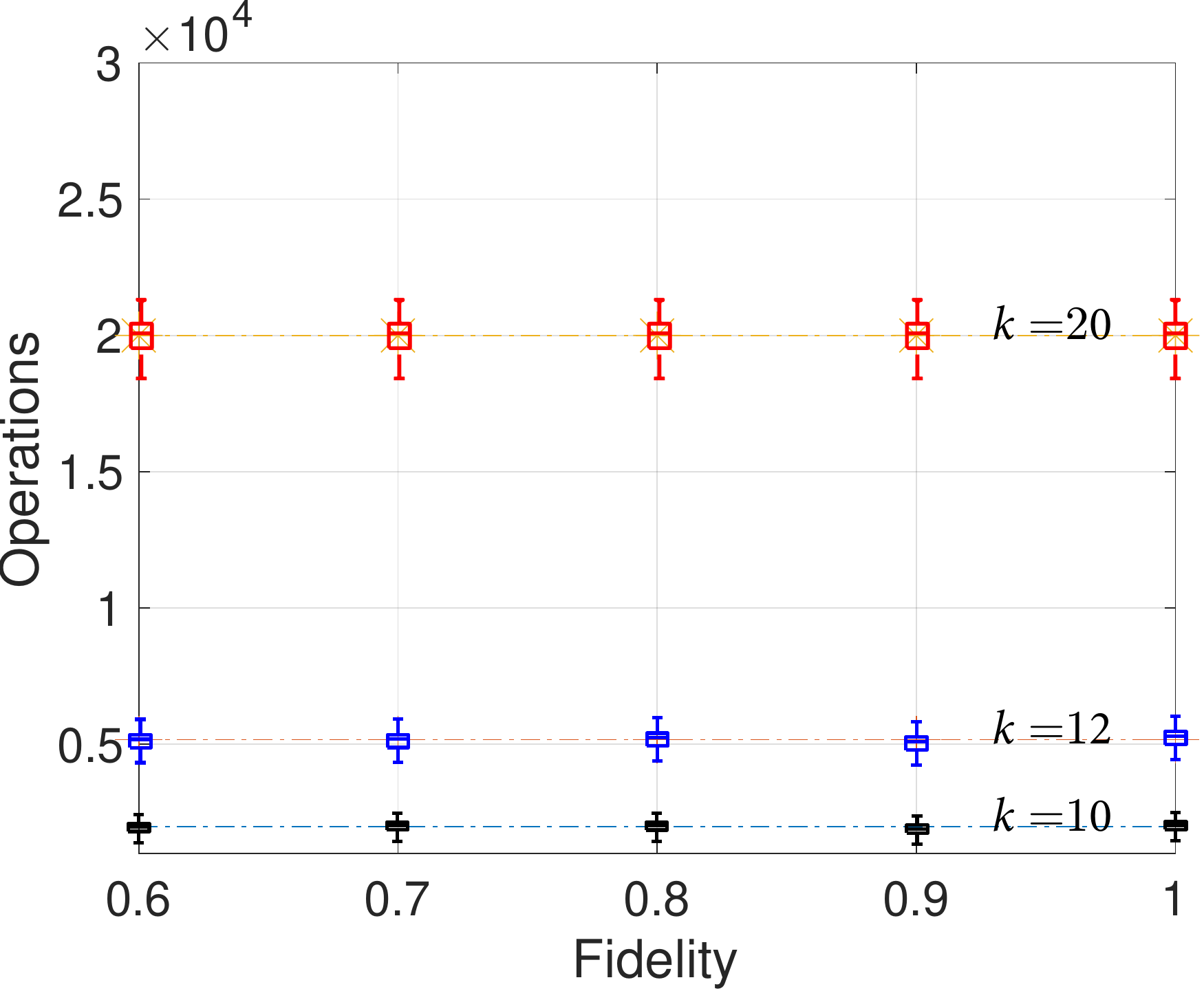}
}
\caption{Congestion setting results in the sparse grid $k$ by $k$ topology, in which $10 \leq k \leq 20$.  (a,c) Required number of physical qubits and operations to ensure entanglement swapping between terminals crossing through the most congested repeater in the network due to repeated purification. (b, d) Required number of physical qubits and operations to ensure entanglement swapping between terminals crossing through the most congested repeater in the network due to concatenated error correction.\label{fig:pythonResults3}}
\end{figure}

As observed above, every order reversal between terminals at the top row creates a path crossing. However, it is also clear that the total number of path crossings depends on the grid topology of the repeaters and on how the paths are chosen and may well exceed the quantity in the right-hand side of Inequality~\eqref{eq:cross}. 

Based on the aforementioned setting, we conduct Monte Carlo 
simulations using the \href{https://networkx.github.io}{NetworkX} python library. The simulation code is available \href{https://github.com/jgalfaro/mirrored-qbcrepgrid/}{online}. In each simulation, congestion is computed as the number of paths crossing through the most visited repeater. The random activation of terminals follows the strategy presented in Figure~\ref{default2}, using $\frac{1}{2}$ as probability $p$ for a terminal at the top of the grid to become active and communicate to some terminal at the bottom row of the grid.  The random arrangement of repeaters and terminals follows the strategy and constraints defined in Section~\ref{sec:netmodel} (i.e., terminal nodes are not adjacent to each other in 
the grid and every terminal node is adjacent to at least one repeater). Figures~\ref{fig:pythonResults2} and~\ref{fig:pythonResults3} show the simulation Results. Every Boxplot corresponds to fifty independent run executions per scenario, increasing the size of the sparse $k$ by $k$ grid, from $k=10$ to $k=20$. We plot the number of required physical qubits and the number of required operations of the most congested repeater (i.e., the one crossed by the higher number of paths in each simulation run). Consistently, we can observe that concatenated error correction, versus repeated purification, presents lower operational complexity than repeated purification to reach high fidelity, at the expense of increasing the number of required physical qubits.

\section{Conclusion}\label{sec:conc}

In a quantum networking environment, we have explored the memory resource requirements analytically and numerically to attain a certain level of fidelity. We have also investigated repeated purification and concatenated error correction in this setting. 

We have observed that concatenated error correction can achieve a given degree of fidelity with fewer iterations than repeated purification, at the cost of considerably increasing the number of required physical qubits. At the same time, we have also observed that the cost in number of operations is higher in the case of repeated purification, compared to concatenated error correction. This results in a comparable amount of resources for both approaches.

As perspectives for future work, one may want to analyze the requirements when combining both techniques simultaneously (concatenated error correction and repeated purification), to estimate the best work memory trade-off while obtaining the highest possible degree of fidelity. 

\bigskip

\noindent \textbf{Acknowledgements ---} We acknowledge the support of the Natural Sciences and Engineering Research Council of Canada (NSERC).

\bibliographystyle{unsrt_MS}

\end{document}